\documentclass{article}
\usepackage{graphicx}
\usepackage{xspace}

\usepackage{preamble}

\newcommand{\fixset}[1]{\text{FIX}(#1)}
\newcommand{\rightside}[1]{\text{RIGHT}({#1})}
\newcommand{\leftside}[1]{\text{LEFT}({#1})}
\newcommand{\tuple}[2]{{\left(\overline{#1},\overline{#2}\right)}}

\newcommand{\FixedSignsTest}[1]{\mathsf{Fixed\ Test}_{#1}}
\newcommand{\SignsTest}[1]{\mathsf{Random\ Signs\ Test}_{#1}}
\newcommand{\CoeffsTest}[1]{\mathsf{Random\ Coefficients\ Test}_{#1}}
\newcommand{\SignsTestInText}{\ensuremath{\mathsf{Random\ Signs\ Test}}\xspace}
\newcommand{\CoeffsTestInText}{\ensuremath{\mathsf{Random\ Coefficients\ Test\ }}}
\newcommand{\GeneratedSubgroupTest}{\mathsf{Generated \ Subgroup \ Test}}

\newenvironment{talign*}
 {\let\displaystyle\textstyle\csname align*\endcsname}
 {\endalign}

\title{Homomorphism Testing  
with Resilience to Online Manipulations}

\author{ { Esty Kelman\thanks{Boston University, Boston, MA, USA, and Massachusetts Institute of Technology, Cambridge, MA, USA.  Email: \texttt{ekelman@mit.edu}} } 
\and {Uri Meir\thanks{Tel-Aviv University, Tel-Aviv, Israel. Email: \texttt{urimeir.cs@gmail.com}}} 
\and {Debanuj Nayak\thanks{Boston University, Boston, MA, USA. Email: \texttt{dnayak@bu.edu}}} 
\and {Sofya Raskhodnikova\thanks{Boston University, Boston, MA, USA. Email: \texttt{sofya@bu.edu}}}}

\date{}

\begin{document}

\maketitle
\begin{abstract}
A central challenge in property testing is verifying algebraic structure with minimal access to data. A landmark result addressing this challenge, the linearity test of Blum, Luby, and Rubinfeld (JCSS `93), spurred a rich body of work on testing algebraic properties such as linearity and its generalizations to low-degree polynomials and group homomorphisms. However, classical tests for these properties assume unrestricted, noise-free access to the input function--an assumption that breaks down in adversarial or dynamic settings. To address this, Kalemaj, Raskhodnikova, and Varma (Theory of Computing `23) introduced the online manipulation model, where an adversary may erase or corrupt query responses over time, based on the tester's past queries.

We initiate the study of {\em manipulation-resilient} testing for {\em group homomorphism} in this online model. Our main result is an {\em optimal} tester that makes $O(1/\varepsilon+\log t)$ queries, where $\varepsilon$ is the distance parameter and $t$ is the number of function values the adversary can erase or corrupt per query. Our result recovers the celebrated $O(1/\varepsilon)$ bound by Ben-Or, Coppersmith, Luby, and Rubinfeld (Random Struct.\ Algorithms `08) for homomorphism testing in the standard property testing model, albeit with a different tester. Our tester, $\mathsf{Random\ Signs\ Test}$, {\em lifts} known manipulation-resilient linearity testers for $\mathbb{F}_2^n\to \mathbb{F}_2$ to general group domains and codomains by introducing more randomness: instead of verifying the homomorphism condition for a sum of random elements, it uses additions and subtractions of random elements, randomly selecting a sign for each element. We also obtain improved group-specific query bounds for key families of groups.
Our results show that despite the challenges of online manipulation, group homomorphism---a fundamental algebraic property---is efficiently testable across a wide range of domains and codomains. Along the way we formalize a general framework for proving resiliency of testers to online manipulations which we believe could be useful in future work.

\end{abstract}

\thispagestyle{empty}
\setcounter{page}{0}

\newpage
\setcounter{page}{1}

\section{Introduction}
Understanding how to verify algebraic structure with minimal access to data is a central question in {\em property testing}. A landmark result addressing this challenge is the linearity test of Blum, Luby, and Rubinfeld~\cite{BLR93}, which checks whether a given function is linear by querying it on two random inputs and their sum. This test and its variants sparked a rich line of work on testing algebraic properties of functions, including linearity 
\cite{BabaiFL91,FeigeGLSS96,BellareCHKS96}, being a low-degree polynomial \cite{BabaiFLS91,GemmellLRSW91,FriedlS95,RubinfeldS96, RazS97, AlonKKLR05, AroraS03, MoshkovitzR08, 
KaufmanR06, Samorodnitsky07, SamorodnitskyT09, JutlaPRZ09, BKSSZ10,HaramatySS13, Ron-ZewiS13, DinurG13,kaufman2022improved}, and being a group homomorphism \cite{BLR93,GrigorescuKS06,Ben-orCLR08,GoldreichR16,BeckerC2022,MooreR2015,gowersH2017,MittalR2024}, with deep connections to probabilistically-checkable proofs \cite{AroraLMSS98, Trevisan98, HastadW03}.

However,  classical property testing, as defined by \cite{RS96, GGR98}, relies on  the assumption of undisrupted query access to the input function—an assumption that may break down in dynamic environments or under adversarial or incomplete access to data.
Many testers, especially for algebraic properties, use structured queries that can be vulnerable to online adversarial interference. To address this, Kalemaj, Raskhodnikova, and Varma~\cite{KalemajRV23} introduced the {\em online manipulation model,} where an adversary can erase or corrupt the input during the execution of the algorithm, based on prior queries. This framework has prompted the reexamination of classical testing problems under a new lens, seeking algorithms that remain correct and efficient despite adversarial interference. 

In this work, we initiate the study of manipulation-resilient testing for group homomorphism, a fundamental property of functions between algebraic structures. Given finite groups $(G,+)$ and $(H,\oplus)$, and a function $f:G\to H$, the goal is to determine whether $f$ is a {\em homomorphism}—i.e., satisfies $f(x_1+x_2)=f(x_1)\oplus f(x_2)$ for all $x_1,x_2\in G$—or is $\eps$-far from every homomorphism.\footnote{As noted in the Encyclopedia of Algorithms \cite{RaskhodnikovaR16}, the set of all homomorphisms from $G$ to $H$ can be viewed as an error-correcting code, known as the \emph{homomorphism code}. 
In the specific setting where $G=\F_2^n$ and $H=\F_2,$ this corresponds to the well-known Hadamard code.
From this perspective, the task of homomorphism testing is equivalent to testing whether a given word is a codeword in the corresponding code.}
The tester accesses the input function $f$ by querying its values $f(x)$ on elements $x$ of the group, but must contend with an {\em online adversary} that knows all queries made by the tester and can manipulate the data over time. Two types of manipulation are considered: an \emph{erasure}, where $f(x)$ is replaced by a special symbol $\perp$, and a \emph{corruption}, where $f(x)$ is replaced by an arbitrary element of $H$. The rate of manipulations is controlled by a parameter $t$. Two types of adversaries have been studied: {\em fixed rate}, which can manipulate $t$ values after every query is answered \cite{KalemajRV23}, and {\em budget-managing}, which accumulates a budget of $t$ manipulations per query and may deploy them later at any point in the computation \cite{BenEliezerKMR24}.

Linearity and low-degree testing are well understood in the online adversarial model. Linearity and quadraticity were among the first properties studied in this setting, with~\cite{KalemajRV23} proving a query lower bound of $\Omega(\log t)$ for testing linearity of functions $f : \F_2^n \to \F_2$.
The linearity tester of \cite{KalemajRV23} was improved by Ben-Eliezer, Kelman, Meir and Raskhodnikova~\cite{BenEliezerKMR24}, who achieved the optimal query complexity of $O(1/\eps + \log t)$, and further refined by Arora, Kelman, and Meir \cite{AroraKM25} to handle the full range of the parameter~$t$. For testing whether a function $f:\F_2^n \to \F_2$ has degree at most $d$, \cite{BenEliezerKMR24} showed a lower bound of $\Omega(\log^d t)$ on query complexity, whereas the elegant tester of Minzer and Zheng \cite{MinzerZ24} works for
functions $f:\mathbb{F}_q^n \to \mathbb{F}_q$ and makes $\log^{O(d)} t / \eps$ queries 
where $q$ is a prime power.
While both lower bounds \cite{KalemajRV23,BenEliezerKMR24} apply to the weakest model of online adversary (fixed-rate with erasures), the final algorithms of~\cite{BenEliezerKMR24, AroraKM25, MinzerZ24} are resilient even in the strongest online adversary model considered (budget-managing with corruptions).

These developments raise a natural question: Can one design testers resilient to online manipulations for the other natural generalization of linearity—namely, group homomorphism?\footnote{It is important to note that the existence of an efficient tester for a property $\cP$ in the standard property testing model does not guarantee that $\cP$ is testable with online adversary. For instance, \cite{KalemajRV23} show that {\em sortedness} of integer sequences—despite having many efficient testers in the standard model \cite{EKKRV00,DGLRRS99,Ras99,BGJRW12,ChakrabartyS13,Raskhodnikova16,Belovs18}, cannot be tested at all even with one online erasure per query (no matter how many queries the tester makes).}
We answer this question affirmatively by designing testers that are resilient to the strongest type of online adversary in our framework—budget-managing with corruptions. Our main result is that homomorphism of functions $f:G\to H$ can be tested with $O(1/\eps+\log t)$ queries in the presence of $t$-online corruption adversary for all groups $G$ and $H$. This bound is optimal: even without manipulations, $\Omega(1/\eps)$ queries are required, and the $\Omega(\log t)$ lower bound of \cite{KalemajRV23} holds for the special case 
$G=\F_2^n$ and $H=\F_2$. 

Since the lower bound applies to specific groups $G$ and $H$, a natural next question is whether some groups admit better testers.
We obtain improved group-specific query bounds
for many important families of groups. In particular, we design an online-manipulation-resilient homomorphism tester that makes $O(1/\eps + E(G))$ queries, where $E(G)$ is the expected number of samples needed to generate $G$.
This tester is optimal, for instance, when $G$ is a cyclic, symmetric, alternating, or simple group.
For vector spaces over finite fields of prime order, we obtain nearly tight bounds: when $G = \F_p^n$ and $H = \F_p^r$, homomorphism can be tested with $O(1/\eps + \log_p t)$ queries, and we prove a matching $\Omega(\log_p t)$ lower bound for the case $r = 1$; when $H = \F_q^r$ where $q \neq p$, testing reduces to checking whether $f$ is identically zero.

Our results demonstrate that group homomorphism—a fundamental algebraic property—admits efficient and manipulation-resilient testers across the board: we give a general optimal tester for all finite groups and complement it with sharper bounds for prominent families of domains and codomains.

\subsection{Our results}\label{sec:results}
Our main result is an {\em optimal} online-manipulations-resilient homomorphism tester of $f:G\to H$ for all groups $G$ and $H$. As all our testers, it works in the strongest online adversary model considered: budget-managing with corruptions.

\begin{theorem}
\label{thm:general_G_H}
    There exists a constant $c \geq 0$ such that for all finite groups $(G,+)$ and $(H,\oplus)$, all $\eps \in (0,1)$ and $t \leq c \cdot \min \set{\eps^2, 1/\log^2 \card{G}} \cdot \card{G}$, there exists an $\eps$-tester for group homomorphism of functions of the form $f:G \to H$
    that works in the presence of every $t$-online erasure (or corruption) budget-managing adversary and makes $O\paren{\frac{1}{\eps} + \log_2 t}$  queries. For erasures, the tester has 1-sided error.
\end{theorem}

Our result recovers the celebrated $O(1/\eps)$ bound by \cite{BLR93,Ben-orCLR08} for testing homomorphism in the standard property testing model, 
 albeit with a different tester. Our tester, which we call 
 $\mathsf{Random}$ $\mathsf{Signs}$ $\mathsf{Test},$
  can be viewed as {\em lifting} the linearity testers developed by \cite{KalemajRV23,BenEliezerKMR24,AroraKM25} for functions $f:\F_2^n\to\F_2$ to work for general groups $G$ and $H$. As noted by \cite{KalemajRV23}, the original \cite{BLR93} tester breaks in the online manipulations model, and as we explain in \Cref{sec:technical-overview}, the known online-manipulation-resilient testers for the case $G=\F_2^n$ and $H=\F_2$ are not sound for general groups. Let $(G,+)$ be the domain group and $(H,\oplus)$ be the codomain group. For the case without manipulations, our tester picks $k$ points (for any even $k$) $x_1,\dots, x_k$ from the domain group $G$ and $k$ signs $\sigma_1,\dots, \sigma_k$ from $\{+,-\}^k$ uniformly and independently at random. Then it obtains an element $a=\sigma_1x_1+\dots+\sigma_k x_k$, where $\sigma_ix_i$ denotes $x_i$ when $\sigma_i=+$ and the inverse of $x_i$ (in $G$) when $\sigma_i=-$.
 Finally, it queries $f$ on $x_1,\dots,x_k$, and $a$, and tests whether the homomorphism condition is satisfied for these elements, that is, whether $f(a) = \sigma_1f(x_1)\oplus\dots\oplus \sigma_kf(x_k))$, where  $\sigma_i = -$  denotes the inverse in the group $H$. 
The manipulation-resilient version of the tester first samples a reserve of elements of $G$ and then runs \SignsTestInText with carefully selected $k$ and elements $x_1,\dots,x_k$ from the reserve . We provide a technical overview of the tester and its analysis in \Cref{sec:technical-overview}. 
\Cref{sec:signs-test} presents and analyzes \SignsTestInText, and 
\Cref{sec:resilience}---a manipulation-resilient version of \SignsTestInText.

\subsubsection{Group-specific homomorphism testing}

Our first group-specific result, \Cref{thm:group-specific-sample-based-erasure-resilient-result},  gives a homomorphism tester  with complexity expressed in terms of the parameter $E(G)$, defined next.

\begin{definition}[\cite{pomerance2002expected}]\label{def:group-parameter-E}
    For all finite groups $(G,+)$, let $E(G)$ be the expected number of independent and uniform samples from $G$ needed to obtain a set that generates $G$. 
\end{definition}

This group specific tester has an additional advantage of being {\em sample-based}, i.e., querying only independent and uniform elements of $G$. Its guarantees are summarized in \Cref{thm:group-specific-sample-based-erasure-resilient-result}, which is proved in \Cref{sec:group-specific-sample-based}.

\begin{theorem}[Group-specific sample-based tester]
\label{thm:group-specific-sample-based-erasure-resilient-result}
    There exists a constant $c > 0$ such that for all finite groups $(G,+)$ and $(H, \oplus)$, all $\eps \in (0,1)$ and  $t \leq c \cdot \min \set{\eps^2, 1/E(G)^2} \cdot \card{G} $, there exists a sample-based $\eps$-tester for group homomorphism of functions of the form $f:G \to H$ that works in the presence of every  $t$-online erasure (or corruption) budget-managing adversary and makes $O(1/\eps + E(G))$ queries. For erasures, the tester has 1-sided error.
\end{theorem}

In the model without manipulations, Goldreich and Ron \cite{GoldreichR16} gave a sample-based homomorphism tester with $O(1/\eps + \log |G|)$ queries and showed this bound is tight for some groups $G$ and~$H$ (in particular, $G = \F_2^n, H = \F_2$). The query bound in \Cref{thm:group-specific-sample-based-erasure-resilient-result} improves on the bound in \cite{GoldreichR16} for many natural groups and is tight when $E(G)$ is constant—e.g., for cyclic groups of prime order $\Z_p$, symmetric groups $S_n$ and alternating groups $A_n$ \cite{dixon1969probability}, and, more generally, finite simple groups \cite{liebeck1995probability}. This gives us the following corollary. 

\begin{corollary}[Homomorphism tester for simple groups]\label{cor:sample-based-abelian} There exists a constant $c>0$ such that for all finite groups $(G,+)$ and $(H,\oplus)$, where $G$ is a finite simple group, all $\eps \in (0,1)$ and all $t\leq c\cdot  \eps^2 |G|$, there exists a sample-based $\eps$-tester for group homomorphism of functions of the form $f:G \to H$ that works in the presence of a $t$-online-erasure adversary and makes $\Theta(1/\eps)$ queries.
\end{corollary}

To improve the bound in \cite{GoldreichR16} and prove \Cref{thm:group-specific-sample-based-erasure-resilient-result}, we analyze the subgroup generated by the sample $S$, 
rather than only partial sums that can be obtained from $S$, obtaining better group-specific bounds on the number of samples needed to properly learn the input homomorphism. 
Our analysis relies on work in group theory~\cite{dixon1969probability,liebeck1995probability,pomerance2002expected,lubotzky2002expected,lucchini2016expected,menezes2013random} that relates $E(G)$ from \Cref{def:group-parameter-E} to the size of the smallest set generating $G$ and other structural parameters of the group.

\paragraph{Bounds for prime fields.}
Finally, we consider groups of the form $G = \F_p^n$ and $H = \F_q^r$, where $p$ and $q$ are primes and $n, r \in \N$. Our next result shows that if $p = q$, the query complexity is below the general bound of $O(1/\eps + \log_2 t)$ from \Cref{thm:general_G_H}.  

\begin{theorem}
\label{thm:prime-fields-p=q}
    There exists a constant $c >0$ such that for all primes $p$, all $n, r \in \N$, all $\eps \in (0,1)$ and $t \leq c \cdot \min \set{\eps^2, 1/n^2} p^n$, there exists an $\eps$-tester for homomorphism of functions of the form $f:\F_p^n \to \F_p^r$ that works in the presence of every $t$-online erasure (or corruption) budget-managing adversary and makes $O(1/\eps + \log_p t)$ queries. For erasures, the tester has 1-sided error.
\end{theorem}

Next we show that, when the range is $\F_p$ (that is, $r=1$), the bound in \Cref{thm:prime-fields-p=q} is tight.

\begin{theorem}
\label{thm:prime-fields-lower-bound}
    For all primes $p$, there exists $n_0 \in \N$ such that for all $n\in \N, n \geq n_0$ and $\eps \in (0, \frac{p-1}{2p}]$, every $t$-online erasure-resilient $\eps$-tester for homomorphism of functions of the form $f:\F_p^n \to \F_p$ must make $\Omega(\log_p t)$ queries.
\end{theorem}

We prove \Cref{thm:prime-fields-p=q,thm:prime-fields-lower-bound} in \Cref{sec:prime-fields}. Finally, we observe that the case when $p\neq q$ is easy.

\begin{rem}
\label{rem: H incompatible wrt G}
    For all finite groups     $G$ and $H$ such that the zero map is the only homomorphism from $G$ to $H$
        (e.g., when $G=\F_p^n$ and $H=\F_q^r$, where $p$ and $q$ are distinct primes and $n,r\in\N$), testing if a function $f: G \to H$ is a homomorphism is equivalent to testing if $f$ is identically zero and can be done with $\Theta(1/\eps)$ queries, even in the presence of a $t$-online manipulation budget-managing adversary.
\end{rem}

\subsubsection{A generic framework for resiliency}
Resiliency to online manipulations calls for testers that query $f$ on elements of $G$ which cannot be easily predicted ahead and tempered with.
This desirable trait of \emph{unpredictability} interferes with the tester's ability to query the most useful (but potentially predictable) elements, which results in the need for more queries overall. 
The tension between the number of queries and how unpredictable they are could be seen as an interpolation between two well-known extremes: query-based testers (in the vanilla testing model), and sample-based testers where each sample is a pair $(x,f(x))$ for a uniformly random $x\in G$, independent of other samples.
Goldreich and Ron~\cite{GoldreichR16} defined a related notion: a tester is $\alpha$-fair if the marginal distribution of each query is $\alpha$-flat, that is, the probability of any element in $G$ to be queried is at most $\alpha$ at any given query.\footnote{To be precise, the parameter $\alpha$ is used a bit differently in \cite{GoldreichR16}, where an $\alpha$-fair tester queries any element with probability at most $1/(\alpha n)$ with $n$ being the size of the domain.
The two notions are equivalent by taking $\alpha' = 1/(\alpha n)$.}
They showed that if a property has an $\alpha$-fair $\eps$-tester that makes $q$ queries with $\alpha = \frac{1}{\Theta(\card{G})}$ and constant $\eps$ then it also has a sample-based tester (in a sense, a ``completely fair'' tester) with $\card{G}^{1-\frac{1}{q}}$ samples.
While this result is an intriguing instance of a trade-off between number of queries and how predictable they are, the notion of fairness seems to be ill-suited for resilience to online manipulations. Indeed, the $3$ queries in the \cite{BLR93} test each have a marginal distribution that is uniformly random over the domain, which makes it fully fair, yet as noted by~\cite{KalemajRV23} it is susceptible to online manipulations as the third query is completely fixed \emph{given the previous two}.

The key to our general framework to prove resiliency is a new notion of unpredictability of a tester (\Cref{def:unpredictable-test}).
Roughly speaking, we say that a tester is $(q,\alpha,\beta)$-unpredictable if, other than low probability $\beta$, all $q$ queries combined are sufficiently unpredictable.
More precisely, we require that for all $i\in[q]$, the $i^{th}$ query is $\alpha_i$-flat with probability at least $1- \beta_i$, where $\alpha = \sum_{i\in[q]}\alpha_i$ and $\beta = \sum_{i\in[q]}\beta_i$.
The sequential and conditional nature is reminiscent of the well-known Santha-Vazirani sources that proved immensely useful in the literature on extractors\cite{santhaV1986}.

We show this notion is well-suited to show resilience to the online manipulation model: no matter which manipulations were made in the past, an unpredictable tester is highly likely to \emph{never} query a manipulated entry.
To be precise, a $(q,\alpha,\beta)$-unpredictable tester has probability of at most $\alpha T + \beta$ of seeing a manipulation, where $T$ is the total amount of manipulated entries by the end of the execution.
To establish our results, we separately prove the soundness and the unpredictability of a core tester, which we then repeat to amplify the success probability to a constant.
We also address some nuances of amplification in the online manipulation model (e.g., later iterations are less reliable due to manipulations) and give an amplification lemma that helps sorting out the choice of certain parameters (the ``size'' of the core algorithm, and the number of repetitions made).

\subsection{Technical overview of the tester for general groups}
\label{sec:technical-overview}
The starting point for the classical homomorphism testers, as well as for manipulation-resilient linearity testers, is a basic test (which can also be described as a local characterization of a property). All basic tests we discuss have perfect completeness, i.e., they always accept every homomorphism. For the standard model, if the basic test has good soundness, it can be repeated to obtain the desired tester. For the online adversarial model, the basic test, in addition to soundness, has to be sufficiently unpredictable in order to evade interference by the adversary. Then (as in \cite{KalemajRV23}) the tester can build a {\em reservoir} by querying uniformly random elements and then use the final query to simulate the basic test on a random subset of the reservoir. Our main technical contribution is the design and analysis of a new basic test. While its full analysis with the reservoir requires additional ideas, we focus here on why a natural extension of the manipulation-resilient linearity test fails for our setting, how our test is constructed, and key insights behind its correctness.

The $(k+1)$-point linearity test of \cite{KalemajRV23} queries the input function $f$ on $k$ independent and uniformly random elements $x_1,\dots,x_{k}$  and their sum $x_{k+1} := \sum_{i\in[k]} x_i$ and checks whether $\sum_{i\in[k]} f(x_i)$ equals $f(\sum_{i\in[k]}x_i)$.
Let $(G,+)$ and $(H,\oplus)$ be finite groups and $f: G\to H$ be a function. A natural extension of the \cite{KalemajRV23} test to the homomorphism property is to check whether\footnote{From now on, we use compact notation to represent our sums: e.g.,
\eqref{eq:k-point-test} can also be written as $\bigoplus_{i \in [k]} f(x_i) = f\Big( \sum_{i \in [k]} x_i\Big)$. In non-Abelian groups, the order matters.
The terms in the sums are ordered according to the index: 
in increasing index order for $\sum_{i\in[k]} x_i$ or $\bigoplus_{i=1}^{k} y_i$, and in decreasing index order for $\sum_{i=k}^{1} x_i$.
}
\begin{equation}
\textstyle{\label{eq:k-point-test}
     f(x_1)\oplus f(x_2)\oplus\cdots\oplus f(x_k)
    \stackrel{?}{=} f(x_1+x_2+\cdots +x_k) }.
\end{equation}

If $f$ is a homomorphism, then \eqref{eq:k-point-test} holds for all tuples $(x_1,\dots,x_{k})$, but the converse is false: for $G = H = \Z_{k-1}$ (the additive cyclic group over $k-1$ elements) and any homomorphism $h:G\to H$, the shifted homomorphism function $f(x) = h(x) \oplus s$ with shift $s\in H$ also satisfies \eqref{eq:k-point-test} for all tuples.\footnote{A similar phenomenon occurs for functions $f:\F_2^n \to \F_2$: the $k$-point tester with even $k$ works for linearity, but odd $k$ produces a test for \emph{affinity} instead.}

To resolve this issue, our homomorphism tester uses a more robust local characterization.

Before introducing our tester, we discuss our desired soundness guarantee and sketch a simple argument showing that many NO instances are rejected by the $(k+1)$-point test with sufficient probability. (In the example above, we saw that this tester fails on functions that are really far from any homomorphism, but it turns out it works well for functions which are closer to a homomorphism).
Let $h_f$ be a closest homomorphism to $f$ and $\eps_f$ be the distance from $f$ to $h_f$. Define $\Delta_f:= \set{x\in G: f(x) \neq h_f(x)}$ (so that $\card{\Delta_f} = \eps_f \cdot \card{G}$).
Since \eqref{eq:k-point-test} holds for $h_f$ with all tuples $(x_1,\dots,x_{k})$, one could analyze how often it holds for $f$ by comparing $f$ and $h_f$ on each tuple. The following ``wishful thinking'' argument is an easy assertion when $\eps_f$ is small enough, and was used before, e.g., for optimal analysis of low-degree testing over $\F_2$~\cite{BKSSZ10}.

\paragraph{Simple case: $\eps_f < \frac 1{2k}.$}
We view each point $x_i$ (including the sum $x_{k+1}$) as a random variable, and ask whether $x_i \in \Delta_f$. Marginally, each $x_i$ is distributed uniformly over $G$ and so this occurs with probability $\eps_f$ for all $i\in[k+1]$.
Intuitively, the points are ``independent enough'' so when $\eps_f < \frac 1{2k}$, w.h.p.\ at most one of them falls in $\Delta_f$, and if exactly one point does - then the tester rejects $f$ (by cancellation rule on the group $H$).
Formally, when $\eps_f < \frac{1}{2k}$ the probability of exactly one $x_i$ falling in $\Delta_f$ can be lower bounded by $k\eps_f/2$ using the premise and pairwise independence of the points and inclusion-exclusion 
(Bonferroni inequality of order $2$),
as was done, e.g., in~\cite{BKSSZ10,BenEliezerKMR24} to analyze low-degree  tests of polynomials over fields.
Note that by a union bound, the rejection probability is at most $k\eps_f$ so the analysis above is optimal (up to a factor of $2$) for inputs with small distance $\eps_f$. 

However, larger values of $\eps_f$ lead to more than a single point in $\Delta_f$ for a typical tuple, and the guaranteed pairwise independence is no longer enough.
Luckily, the $k+1$ points used in \eqref{eq:k-point-test} are in fact $k$-wise independent, which is a stronger guarantee for $k > 2$, and we leverage this fact.

\paragraph{The elementary argument.}
Let $p_m$ denote the probability that $m$ independent and uniformly random elements $z_1,\dots, z_m \in G$ satisfy $\bigoplus_{i\in[m]} f(z_i) = \bigoplus_{i\in[m]} h_f(z_i)$, and let $\gamma_m = 1 - p_m$ be the probability of an inequality.
By definition, $\gamma_1 = \eps_f$.
We utilize the $k$-wise independence by splitting the last point, $x_{k+1}$, from the rest and comparing the RHS and LHS of \eqref{eq:k-point-test} with $f$ to those with $h_f$ where the equality holds.
In particular, the event that the RHS of $f$ and $h_f$ agree but the LHS are different, leads to rejection on $f$. Then
\begin{talign*}
     &\Pr_{x_1,\dots,x_k\in G}{\Big[\bigoplus_{i\in[k]}f(x_i) \neq f\Big(\sum_{i\in[k]} x_i\Big) \Big]}\\
    \geq &\Pr_{x_1,\dots,x_k\in G}{\Big[\bigoplus_{i\in[k]}f(x_i) \neq \bigoplus_{i\in[k]}h(x_i)\ \bigwedge\  f\Big(\sum_{i\in[k]} x_i\Big) = h\Big(\sum_{i\in[k]} x_i\Big)\Big]}\\
    \geq &\Pr_{x_1,\dots,x_k\in G}{\Big[\bigoplus_{i\in[k]}f(x_i) \neq \bigoplus_{i\in[k]}h(x_i)\ \Big]}
    - \Pr_{x_1,\dots,x_k\in G}{\Big[f\Big(\sum_{i\in[k]} x_i\Big) \ne h\Big(\sum_{i\in[k]} x_i\Big)\Big]}\\
    = &\gamma_k - \gamma_1 ,
\end{talign*}
where the second inequality uses $\Pr[A\wedge B] \geq \Pr[A] - \Pr[\bar{B}]$
for two events $A, B$ (where $\bar{B}$ is the complement of $B$). The equality is since $(x_1,\dots,x_k)$ are i.i.d. samples taken uniformly at random from $G$.

From here on, the goal is to lower bound $\gamma_{k}$. To showcase the argument we restrict ourselves to $H = \F_2$ (which in particular captures linearity over $\F_2$), and discuss the extension later.
Here, $f(x)-h_f(x) = 1$ if $x\in\Delta_f$, and otherwise $f(x) = h_f(x)$.
That is, $\gamma_{k}$ is the probability that there is an odd number of indices $i\in[k]$ for which $x_i\in\Delta_f$.
Since these points are independent and each has probability $\eps_f$ to fall inside $\Delta_f$, this equals to the probability of a binomial random variable $Z\sim Bin(k, \eps_f)$ being odd, which is known to equal $\frac{1 - (1-\eps_f)^{k}}{2}$ (see \Cref{fact:binomial rv is even}).
This expression has a known lower bound of $\Omega(\min\set{k\eps_f , 1})$, which absorbs the substraction of $\gamma_1 = \eps_f$ for any $\eps_f \leq 1/8$ and completes the proof.\footnote{We use a precise version $\frac{1 - (1-\eps_f)^{k}}{2} = \Omega(\min\set{k\eps_f,1})$ with particular constant, by applying~\cite[Lemma 3.1]{BenEliezerKMR24}.}

To generalize this argument, note that $H = \F_2$ is essentially the worst-case scenario, as ``coincidental cancellations" are most probable when $f(x) - h_f(x) \in \set{0,1}$.
In section \Cref{sec:soundness-of-signs-tester-Small-distance} we prove the claim for our modified tester, described shortly, and any codomain $H$, by analyzing an inductive recurrence directly on $p_m$ (rather than $\gamma_m = 1 - p_m$).

\paragraph{The random signs test.} We make the local characterization more robust by adding \emph{signs}.
That is, we rely on the inverse operation that is guaranteed in every group, and use $-x$ to express addition of the inverse of $x$, for both groups $G$ and $H$.
Since every homomorphism satisfies $f(-x) = -f(x)$ for all $x\in G$, it respects not only summations of elements but signed summations as well.
The corresponding tester is this: given a function $f: G \to H$,  choose uniformly at random (and independently) a tuple $(x_1,\dots,x_{k},\sigma_1,\dots,\sigma_{k})$ consisting of $k$ elements $x_1,\dots,x_k \in G$, and $k$ signs $\sigma_1,\dots,\sigma_{k}\in\set{+,-}$. Finally, check wether
\begin{equation}
     \label{eq:signed-k-point-test}
    \bigoplus_{i\in[k]} \sigma_i f(x_i) \stackrel{?}{=} f(\sum_{i\in[k]}\sigma_i x_i) .
\end{equation}
As before, if $f$ is a homomorphism, all tuples satisfy equation~\eqref{eq:signed-k-point-test}. The difference is that now, when $k$ is even, the converse holds as well.\footnote{Indeed, if $f$ always satisfies equation~\eqref{eq:signed-k-point-test}, then given $x_1,x_2\in G$ we use the tuple with elements $x_1,x_2$ and $x_j = 0$ for $j\geq 3$, and signs $\sigma_1 = \sigma_2 = '+'$, and $\sigma_j$ for $j\geq 3$ alternates between addition and subtraction. The $k-2$ alterations of $+0, -0$ cancel out on the RHS, and similarly the alterations
 of $+f(0)-f(0)$ cancel out on the LHS (even if $f(0) \neq 0_H$), which asserts that $f(x_1) + f(x_2) = f(x_1+x_2)$. This holds for any $x_1,x_2\in G$, showing that $f$ is a homomorphism.}
To analyze the soundness of \SignsTestInText, we consider two cases. For $\eps_f < 1/8$, the elementary argument (\cref{lem:soundeness-signs-test-small-eps-general-H}) suffices.
To prove this tester rejects functions with $\eps_f \geq 1/8$ (unlike the unmodified $k$-point tester), we show the contrapositive via a corrector argument based on that of \cite{Ben-orCLR08}.

In the corrector method, a corrector function $g(x)$ is defined as the value satisfying the most instances of equation~\eqref{eq:signed-k-point-test}, where each instance corresponds to a signed tuples for which $\sum_{i\in[k]} \sigma_i x_i = x$.
Differently put, each tuple of elements and signs that sum to $x$ "votes" for a value it wishes $f(x)$ to be --- then $g(x)$ is defined as the plurality of these votes.
To prove soundness, one shows that if the test rejects $f$ with probability $\mu \leq 1/10$, then two key assertions follow: (a) $f$ is $2\mu$-close to $g$; (b) $g$ is a homomorphism.
The two assertions are proved in \Cref{lem:corrector_bounding_delta} and \Cref{lem:corrector_g_homomorphisms} respectively, with the latter employing the probabilistic method.
Both rely on an auxiliary step, \Cref{lem:corrector_bounding_eta}, which shows that the choice of $g(a)$ as a plurality of votes, is in fact a strong majority of at least $1-2\mu$ of the votes, for every input $a\in G$.

We focus here on the proof of this step, which demonstrates the techniques of the full argument, and the role of random signs.
The proof analyzes the collision probability --- the probability that two signed tuples $(\sigma,x), (\tau,y)$ that sum to $a$ vote for the same value for $f(a)$.
By re-arranging, we get
\[
        \textstyle{
        \underbrace{\bigoplus_{i \in [2k]} \sigma_i f(x_i)}_{S_{\sigma,x}}
        = 
        \underbrace{\bigoplus_{i \in [2k]} \tau_i f(y_i)}_{S_{\tau,y}}
        \iff 
        \underbrace{\bigoplus_{i=0}^{k-1}
        -\tau_{k-i} f(y_{k-i}) \oplus
        \bigoplus_{i\in[k]} \sigma_i f(x_i)}_{S_1} 
        = 
        \underbrace{\bigoplus_{i=k+1}^{2k} \tau_i f(y_i) \oplus \bigoplus_{i=0}^{k-1} -\sigma_{2k-i} f(x_{2k-i})}_{S_2}.
        }
    \] 
Crucially, the tuples underlying the sums $S_1, S_2$ are each a mix from $(\sigma,x)$ and $(\tau,y)$ and we show that both have the distribution of the random sign test, but they sum (in $G$) to the same element, $a'$.
Therefore 
\[
        \textstyle{
        \Pr{[S_{\sigma, x} \neq S_{\tau, y}]}
        = \Pr{[S_1 \neq S_2]}
        \leq \Pr{[S_1 \neq f({a'}) \vee S_2 \neq f({a'})]}
        \leq \Pr{[S_1 \neq f({a'})]} + \Pr{[S_2 \neq f({a'})]} = 2\mu .}
    \]
The necessity of random signs in the technique is apparent: without them, the sums $S_{\sigma, x}, S_{\tau,y}, S_1, S_2$ correspond to tuples with different sign vector, failing to relate the probability of this event to $\mu$ (the rejection probability of the test).
By adding random signs, 
each element has a symmetric role in the test, which leaves the connection between test and equation even after one rearranges the elements, which is a delicate matter especially in non-abelian groups. 
Similar issues arise for odd $k$,
as there is no way to split the two summations ``down the middle''.
Both technical requirements appear to be tight: we have seen a choice for $G,H$ where fixed signs (say, all are $+$) fail to characterize homomorphisms.
Yet for $G = \F_2^n, H = \F_2$ the signs have no meaning ($x = -x$ always), and our test coincides the simple $k$-point test, which requires an even value of $k$ to test linearity (otherwise, it tests affinity).

\subsection{Related work}\label{sec:related_work}
\paragraph{Injecting more randomness vs.\ derandomization.}
Interestingly, while a significant amount of work \cite{Trevisan98,SamorodnitskyT00,HastadW03,Ben-SassonSVW03,ShpilkaW06} went into reducing the number of random bits used by homomorphism tests, the online manipulation model presents a contrasting goal: increasing randomness to make future queries less predictable to the adversary. Indeed, Kelman, Linder and Raskhodnikova~\cite{KelmanLR25} show that for certain properties, testing in the online model requires exponentially more random bits than in the offline setting (even with erasures), despite comparable query complexity. This increased randomness requirement motivates the design of testers that are more {\em sample-based}, i.e., use uniformly random labeled samples in place of some arbitrary (potentially adaptive) queries.
This trend aligns with learning-theoretic settings, where labeled samples are typically modeled as less costly than arbitrary queries.

\paragraph{Low-error regime.}
Different versions of $k$-point testers recently appear in the context of testing in the $1\%$ regime, in the classic model (with no online manipulations).
A recent work by Khot and Mittal \cite{KhotM2025} gives testing results for linearity over the \textit{biased} hypercube.
A work concurrent to ours by Mittal and Roy \cite{MittalR2025} presents a general framework for homomorphism testing in the $1\%$ regime, and apply it to various choices of $G, H$.
Specifically for the case of vector spaces over prime fields, \cite{MittalR2025} show a $k$-point test with random coefficients, identical to the test we devise specifically for such choice of $G,H$ (\Cref{alg: random coeffs test}).
Curiously, the exact same algorithm is well-suited for both tasks ($1\%$ regime in the classic model, and $99\%$ in the online manipulations model).

The $k$-point testers of both work do not choose their points uniformly at random, but rather according to "importance". The former according to the weight of each point in the biased hypercube, and the latter chooses a $k$-tuple from $G$ according to a newly defined measure of its "importance" in the structure of the group (which roughly speaking comes from the order of the elements in the group $G$).
This sophisticated choice of $k$-tuples makes it unclear whether such strategies are resilient to online manipulations, as an adversary is prone to manipulate the important data points the test is likely to use.

\subsection{Open questions}\label{sec:open}
Our main results are in a sense orthogonal: \Cref{thm:general_G_H} applies to all finite groups $G$ and $H$, regardless of their structure, precisely quantifying the dependence  on $t$.
In contrast, \Cref{thm:group-specific-sample-based-erasure-resilient-result} shows that $E(G)$ samples, rather than $\log\card{G}$ samples, suffice to generate the entire domain group $G$,and succesfully apply the learn-and-test approach.
Combining both aspects should yield stronger results in the online manipulation model, as in the case of finite fields (\Cref{thm:prime-fields-p=q}). More results of this flavor would be interesting, focusing on the following question: given a group $G$ and $t\in\N$, how many random samples from $G$ are required to generate w.h.p.\ a subgroup $G' \subseteq G$ of size $\card{G'}$ that exceeds $t$?

A different avenue is utilizing the structure of $H$ and its \emph{compatibility} to $G$, both for upper and lower bounds.
As mentioned in \Cref{rem: H incompatible wrt G}, $H$ could reduce the complexity dramatically when it is incompatible with $G$.
Is it possible to define a measure of compatibility that dictates the complexity of testing homomorphism from $G$ to $H$?
For lower bounds, could we find for each $G$ a group $H$ so that \Cref{thm:group-specific-sample-based-erasure-resilient-result} is tight? One candidate would be $H = Aut(G)$, the group of automorphisms from $G$ to $G$.

\paragraph{Organization.}
The rest of the paper is organized as follows. We begin in \Cref{sec:prelims} with formal definitions of the problem and the online manipulation model. In \Cref{sec:signs-test}, we introduce our core technical contribution, the $\mathsf{Random\ Signs\ Test}$ (\Cref{alg: random signs test}), and analyze its soundness in the standard (non-adversarial) model. In \Cref{sec:general_framework}, we present the general framework for achieving online manipulation resilience, defining 'unpredictable testers' and proving our generic amplification lemma (\Cref{lem: generic amplification}). We then apply this framework in \Cref{sec:resilience} to prove our main optimal bound for general groups (\Cref{thm:general_G_H}). Finally, \Cref{sec:group-specif} is dedicated to our group-specific results, where we prove the bounds for sample-based testers (\Cref{thm:group-specific-sample-based-erasure-resilient-result}) and for functions over prime fields (\Cref{thm:prime-fields-p=q} and \Cref{thm:prime-fields-lower-bound}).

\section{Preliminaries}
\label{sec:prelims}
For two groups $(G,+)$ and $(H,\oplus),$ the function $h:G \to H$ is a  group {\em homomorphism} 
if $h(x_1 + x_2) = h(x_1)\oplus h(x_2)$ for all $x_1,x_2 \in G$.
In non-Abelian groups, the order matters.
The terms in the sums are ordered according to the index: 
in increasing index order for $\sum_{i\in[k]} x_i$ or $\bigoplus_{i=1}^{k} y_i$, and in decreasing index order for $\sum_{i=k}^{1} x_i$.
The set of all group homomorphisms from $G$ to $H$ is denoted $\mathsf{HOM}(G, H)$. The (relative) Hamming {\em distance} between two functions $f,g:G\to H$ is $dist(f,g):= \Pr_{x\in G} [f(x)\neq g(x)].$ The {\em distance from $f:G\to H$ to a property} (i.e., a set) $\cP$ of functions of the same form is $dist(f,\cP):=\inf_{g\in\cP}dist(f,g).$
For every function $f:G\to H$, let $\eps_f:=dist(f,\mathsf{HOM}(G, H))$. A function $f$ is {\em $\eps$-close to a homomorphism} if $\eps_f\leq \eps$ and is {\em $\eps$-far from a homomorphism} otherwise.
\begin{definition}[Online $\eps$-tester \cite{KalemajRV23,BenEliezerKMR24}]
\label{def:online_tester}
Fix $\eps\in(0,1).$
    An online $\eps$-tester $\cT$ for a property~$\mathcal{P}$ that works in the presence of a specified adversary (e.g., $t$-online erasure budget-managing) is given access to an input function $f$ via oracle access to a related function $f'$ that is initially equal to $f$ and, after each query answered, may be modified by the adversary.
    For all adversarial strategies of the specified type,
    \begin{enumerate}
        \item if $f \in \mathcal{P}$, then $\cT$ accepts with probability at least 2/3, and
        
        \item if $f$ is $\eps$-far from $\cP,$
        then $\cT$ rejects with probability at least 2/3, 
    \end{enumerate}
   where the probability is over the internal randomness of $\cT$. 
      If $\cT$ works in the presence of an erasure (resp., corruption) adversary, we refer to it as an online-erasure-resilient (resp., online-corruption-resilient) tester.
       If $\cT$ always accepts every function $f\in\cP$, then $\cT$ has \emph{1-sided error.}
        If all queries of $\cT$ are chosen uniformly and independently from the domain of $f$, then $\cT$ is {\em sample-based}.
    \end{definition}

\section{Random Signs Test}\label{sec:signs-test}
In this section, we present our \SignsTestInText (\Cref{alg: random signs test}) for testing group homomorphism. This test always accepts functions that are homomorphisms because it only rejects if a violation of the homomorphism property is found. We analyze the soundness of the test in \Cref{thm:general-soundness-based-on-eps}.

\begin{algorithm}
\caption{$\SignsTest{k}$}
\label{alg: random signs test}
\begin{algorithmic}[1]
\Ensure $k \in \N$
\Require Query access to a function $f: G \to H$ where $(G,+)$ and $(H,\oplus)$ are finite groups.
\State Draw $k$ independent and uniformly random elements
$x_1, x_2, \dots, x_{k}$ from $G$.
\State Draw $k$ independent and uniformly random 
signs $\sigma_1, \dots, \sigma_{k}$ from $\set{+, -}$.
\State\label{step:query} Query $f(x_1), f(x_2), \dots, f(x_{k})$ and
$f(a)$, where  $a \gets \sum_{i\in[k]} \sigma_i x_i.$ 
\Comment{$+b=b$ and $-b$ is the inverse of $b$}
\State\label{step:signs-condition}\textbf{Accept} if $\bigoplus_{i\in [k]} \sigma_i f(x_i) = f(a)$; otherwise, \textbf{reject}. \Comment{for both $b\in G$ and $b\in H$.}
\end{algorithmic}
\end{algorithm}

\begin{theorem}
\label{thm:general-soundness-based-on-eps}
    For all  $k\in\N$, $\eps \in (0,1)$, finite groups $(G,+)$ and $(H,\oplus)$,   and functions $f:G \to H$ such that $f$ is $\eps$-far from being a homomorphism,
        \begin{equation}
         \Pr{[\SignsTest{2k} \text{ rejects}]} \geq \min\set{\frac{(2k-3)\cdot \eps}{3}, \frac{1}{16}}.
    \end{equation}
\end{theorem}
Recall that $\eps_f$ denotes the (relative) distance from $f$ to $\mathsf{HOM}(G, H)$. We analyze the soundness of \SignsTestInText for small $\eps_f$ in \Cref{sec:soundness-of-signs-tester-Small-distance} and for large $\eps_f$ in \Cref{sec:soundness-of-signs-tester-Large-distance}. We put these two results together and complete the proof of \Cref{thm:general-soundness-based-on-eps} in \Cref{sec:completing-the-soundness-analysis}.

\subsection{Soundness analysis for small \texorpdfstring{$\eps_f$}{epsilon-f}}\label{sec:soundness-of-signs-tester-Small-distance}

In this section, we state and prove \Cref{lem:soundeness-signs-test-small-eps-general-H}, establishing the soundness of \Cref{alg: random signs test} for $\eps_f < \frac{1}{8}$.

\begin{lemma}
\label{lem:soundeness-signs-test-small-eps-general-H}
    For all $k\in\N$, finite groups $(G,+)$ and $(H,\oplus)$, and functions $f: G \to H$ 
    with $\eps_f< \frac 1 8$,
        \begin{equation}
    \label{eq:test soundness 2}
         \Pr{[\SignsTest{k} \text{ rejects}]} \geq \min\set{\frac{(k-3)\cdot \eps_f
                  }{3}, \frac{1}{6}}.
    \end{equation}
\end{lemma}

To  prove \Cref{lem:soundeness-signs-test-small-eps-general-H}, we
introduce  $\FixedSignsTest{k}$ (\Cref{alg: fixed signs test}) which uses a fixed sequence of $k$ signs $\overline{\sigma} = (\sigma_1, \sigma_2, \dots, \sigma_k)$, in contrast to the uniform and independent signs from ${+,-}$ used in \Cref{alg: random signs test}.

\begin{algorithm}
\caption{$\FixedSignsTest{k}(\overline{\sigma})$}
\label{alg: fixed signs test}
\begin{algorithmic}[1]
\Ensure $k \in \N$, sign sequence $\overline{\sigma} \in \set{+,-}^k$
\Require Query access to a function $f: G \to H$ where $(G, +)$ and $(H, \oplus)$ are finite groups.
\State Draw $k$ independent and uniformly random elements
$x_1, x_2, \dots, x_{k}$ from $G$.
\State Query $f(x_1), f(x_2), \dots, f(x_{k})$ and
$f(a)$, where  $a \gets \sum_{i\in[k]} \sigma_i x_i.$
\State \textbf{Accept} if $\bigoplus_{i \in [k]} \sigma_i f(x_i) = f(a)$; otherwise, \textbf{reject}.
\end{algorithmic}
\end{algorithm}

Note that \Cref{alg: fixed signs test} with $\overline{\sigma}$ set to $+^k$ recovers the tester of
\cite{KalemajRV23}.

\begin{lemma}
\label{lem:soundeness-fixed-signs-test-small-eps-general-H}
    For all $k\in\N$, all sign sequences $\overline{\sigma} \in \set{+,-}^k$, finite groups $(G,+)$ and $(H,\oplus)$, and functions $f: G \to H$ 
    with $\eps_f< \frac 1 8$,
        \begin{equation}
    \label{eq:test soundness 3}
         \Pr{[\FixedSignsTest{k}(\overline{\sigma}) \text{ rejects}]} \geq \min\set{\frac{(k-3)\cdot \eps_f
                  }{3}, \frac{1}{6}}.
    \end{equation}
\end{lemma}

We prove \Cref{lem:soundeness-fixed-signs-test-small-eps-general-H} after showing that it implies \Cref{lem:soundeness-signs-test-small-eps-general-H}.

\begin{proof}[Proof of \Cref{lem:soundeness-signs-test-small-eps-general-H} assuming \Cref{lem:soundeness-fixed-signs-test-small-eps-general-H}]
$\SignsTest{k}$ (\Cref{alg: random signs test}) chooses the signs $\overline{\sigma}$ independently and uniformly at random from $\set{+,-}^k$. Thus, the rejection probability of $\SignsTest{k}$ is the expected rejection probability of $\FixedSignsTest{k}(\overline{\sigma})$ for uniformly random $\overline{\sigma}$.
\end{proof}

\begin{proof}[Proof of \Cref{lem:soundeness-fixed-signs-test-small-eps-general-H}]
    Fix a sign sequence $\overline{\sigma}$ in $\set{+,-}^k$. Let $g:G\to H$ be a closest homomorphism 
        to $f$, i.e., $\Pr_{x\in G}{[f(x) \neq g(x)]}= \eps_f$. 
    Then     $\FixedSignsTest{k}(\overline{\sigma})$ always accepts $g$, since for all all $x_1, x_2, \dots, x_k \in G$,
    \begin{equation}
    \label{eq:g always passes test}
         \bigoplus_{i\in[k]}\sigma_ig(x_i) = g\Big(\sum_{i\in[k]} \sigma_i x_i\Big).
     \end{equation}
        The core of the analysis is to compare \Cref{alg: fixed signs test} behaves         on $f $ and $g$ for each sample.                 One way \Cref{alg: fixed signs test} can reject $f$ is 
                by selecting $x_1, \dots, x_k$ such that the right-hand side (RHS) of \eqref{eq:g always passes test} remains unchanged when replacing $g$ with $f$, but the left-hand side (LHS) does not.
     Therefore,
    \begin{align}
    \label{eq:comparison-of-g-and-f}
        \Pr{[\FixedSignsTest{k}(\overline{\sigma}) \text{ rejects}]} \geq \Pr_{x_1,\dots,x_k\in G}{\Big[\bigoplus_{i\in[k]}\sigma_if(x_i) \neq \bigoplus_{i\in[k]}\sigma_ig(x_i)\ \bigwedge\  f\Big(\sum_{i\in[k]} \sigma_i x_i\Big) = g\Big(\sum_{i\in[k]} \sigma_i x_i\Big)\Big]}.
    \end{align}
        We consider the probabilities of the  complements of the two events---on the LHS and the RHS of \eqref{eq:comparison-of-g-and-f}---obtaining the complement of the event on the RHS by applying the De Morgan's law:
        \begin{align}
        \Pr&{[\FixedSignsTest{k}(\overline{\sigma}) \text{ accepts}]}
        \leq 
        \Pr_{x_1,\dots ,x_k\in G}{\Big[\bigoplus_{i\in[k]}\sigma_if(x_i) = \bigoplus_{i\in[k]}\sigma_ig(x_i)\ \bigvee\ f\Big(\sum_{i\in[k]} \sigma_ix_i\Big) \neq g\Big(\sum_{i\in[k]} \sigma_ix_i\Big)\Big]}\nonumber\\
        &\leq \Pr_{x_1,\dots ,x_k\in G}{\Big[\bigoplus_{i\in[k]}\sigma_if(x_i) = \bigoplus_{i\in[k]}\sigma_ig(x_i)\Big]} + \Pr_{x_1,\dots ,x_k\in G}{\Big[f\Big(\sum_{i\in[k]} \sigma_ix_i\Big) \neq g\Big(\sum_{i\in[k]} \sigma_ix_i\Big)\Big]} ,\label{eq:two-terms}
    \end{align}
    where \eqref{eq:two-terms} holds by a union bound.
        The second term of \eqref{eq:two-terms} is exactly  $\eps_f$, as $\sum_{i \in [k]} \sigma_i x_i$ is a uniformly random element from $G$, as shown in \Cref{clm: sum is random}. Let $p_k$ denote the first term of \eqref{eq:two-terms} where the subscript $k$ indicates the number of elements the probability is taken over. 
        
        We next use induction on $k$ to prove the following: Let $k_0 = \lfloor 1/\eps_f \rfloor$, then
    \[
    p_k \le \begin{cases}
        1 - k \eps_f / e & \text{for } k \in [k_0]; \\
        \max\{p_{k_0}, 1/2 + \eps_f\} & \text{for } k > k_0.
    \end{cases}
    \]

    \begin{claim}[Base Case]
    \label{clm: rejection based on one point}
        For all $k \le 1/\eps_f$, we have $p_k \le 1 - k \eps_f / e$.
    \end{claim}
    \begin{proof}
        We analyze the probability of the complement:
        \begin{align*}
            1 - p_k &= \Pr_{x_1, \dots, x_k \in G} \left[ \bigoplus_{i\in[k]} \sigma_i f(x_i) \ne \bigoplus_{i\in[k]} \sigma_i g(x_i) \right] 
            \\&
            \ge \Pr \big[\sigma_i f(x_i) \ne \sigma_i g(x_i) \text{ for exactly one } i \in [k] \big] 
                        = k \eps_f (1 - \eps_f)^{k - 1} \ge \frac{k \eps_f}{e},
        \end{align*}
        where the second to last inequality holds because $\sigma_i f(x_i) \neq \sigma_ig(x_i)$ iff $f(x_i) \neq g(x_i)$ and the last inequality holds since $(1 - \eps_f)^{k - 1} \ge  (1 - \eps_f)^{1/\eps_f - 1}$ for $k \le 1/\eps_f$ and $(1 - \eps_f)^{1/\eps_f - 1} \geq 1/e$ for $\eps_f \in (0, 1]$.
    \end{proof}

   \begin{claim}[Inductive Step]
   \label{clm:p_k cannot increase too much}
   For all $k>1$, we have  $p_k\le\max\{p_{k-1},1/2+\eps_f\}$.
    \end{claim}
\begin{proof} We represent the event that defines $p_k$ as a union of two disjoint events:
    \begin{align*}
        p_{k}&=\Pr_{x_1,\dots ,x_k\in G}{\Big[ \bigoplus_{i\in[k]}\sigma_i f(x_i) = \bigoplus_{i\in[k]}\sigma_i g(x_i) \Big]}\\
        &= \Pr_{x_1,\dots ,x_k\in G}{\Big[ \bigoplus_{i \in [k-1]}\sigma_i f(x_i) = \bigoplus_{i\in[k-1]}\sigma_i g(x_i) \bigwedge \sigma_k f(x_k)= \sigma_k g(x_k) \Big]}\\
        &
        \qquad+\Pr_{x_1,\dots ,x_k\in G}{\Big[ \bigoplus_{i\in[k-1]}\sigma_i f(x_i) \neq \bigoplus_{i\in[k-1]}\sigma_i g(x_i)\bigwedge \sigma_k f(x_k) = \bigoplus_{i \in [k-1]} (-\sigma_{k-i}f(x_{k-i})) \oplus \bigoplus _{i\in[k]}\sigma_i g(x_i) \Big]}.
        \end{align*}
        The first summand is         $p_{k-1}(1-\eps_f)$ by the independence of $x_i$'s and the fact that $\sigma_k f(x_k) = \sigma_k g(x_k)$ iff $f(x_k) = g(x_k)$.
        For the second summand, we have: 
        \begin{align}
            &\Pr_{x_1,\dots ,x_k\in G}{\Big[\bigoplus_{i\in[k-1]}\sigma_i f(x_i) \neq \bigoplus_{i\in[k-1]}\sigma_i g(x_i)\bigwedge \sigma_k f(x_k)= \bigoplus_{i \in [k-1]} (-\sigma_{k-i}f(x_{k-i})) \oplus \bigoplus_{i\in[k]}\sigma_i g(x_i) \Big]}\nonumber\\
            &\quad =\Pr{\Big[ \bigoplus_{i\in[k-1]}\sigma_i f(x_i) \neq \bigoplus_{i\in[k-1]}\sigma_i g(x_i) \Big]} \nonumber\\
            &\quad\quad\times \Pr{\Big[ \sigma_k f(x_k) = \bigoplus_{i \in [k-1]} (-\sigma_{k-i}f(x_{k-i}))  \oplus \bigoplus_{i\in[k]}\sigma_i g(x_i)|\bigoplus_{i\in[k-1]}\sigma_i f(x_i) 
            \neq \bigoplus_{i\in[k-1]}\sigma_i g(x_i) \Big]} \label{eq:small-e_f-summand}\\
            &\quad \le (1-p_{k-1})\eps_f,\label{eq:small-e_f-final-inequality}
        \end{align}
        where the inequality in \eqref{eq:small-e_f-final-inequality} holds, since $\displaystyle \bigoplus_{i\in[k-1]}\sigma_i f(x_i) \neq \bigoplus_{i\in[k-1]}\sigma_i g(x_i)$ implies $\displaystyle\bigoplus_{i \in [k-1]} (-\sigma_{k-i}f(x_{k-i}))\oplus \bigoplus_{i\in[k]}\sigma_i g(x_i)\neq \sigma_k g(x_k)$, and thus \eqref{eq:small-e_f-summand} is at most $\Pr[\sigma_k f(x_k)\neq \sigma_k g(x_k)]=\eps_f$.
        Combining all the above, we conclude that 
        \[p_k\leq p_{k-1}(1-\eps_f)+(1-p_{k-1})\eps_f=p_{k-1}+\eps_f(1-2p_{k-1})\leq \max\{p_{k-1},1/2+\eps_f\}.\]
        The first bound in the last inequality holds when $p_{k-1}\ge 1/2$ and the second otherwise.
        \end{proof}

Note that $k_0=\lfloor 1/\eps_f\rfloor> 1/\eps_f-1$. Using \Cref{clm: rejection based on one point} we have $p_{k_0}\leq 1-k_0\eps_f/e< 1-(1-\eps_f)/e$. If $\eps_f \leq 1/8$, then $p_{k_0} \leq 0.7$. Using \Cref{clm:p_k cannot increase too much} we have for all $k\geq k_0$, $p_k\le \max\{p_{k_0},1/2+\eps_f\} \leq 0.7$ for $\eps_f \le 1/8$.  We complete the proof by noting that $\Pr{[\FixedSignsTest{k}(\overline{\sigma}) \text{ accepts}]} \leq p_k + \eps_f$ from \eqref{eq:two-terms}. Thus, $\Pr{[\FixedSignsTest{k}(\overline{\sigma}) \text{ rejects}]}$  $\ge (1-p_k) -\eps_f \geq \min\{k\eps_f/e-\eps_f,0.3-\eps_f\}\ge\min\{(k-3)\eps_f/3,1/6\}$.
\end{proof}

\subsection{Soundness analysis for large \texorpdfstring{$\eps_f$}{epsilon-f}}\label{sec:soundness-of-signs-tester-Large-distance}

In this section, we show that \SignsTestInText used with even $k$ provides soundness proportional to~$\eps_f$. Specifically, we prove the following lemma.

\begin{lemma}
\label{lem:large distance soundness guarantee}
    For all $k\in\N$, finite groups $(G,+)$ and $(H,\oplus)$,  and     functions $f:G \to H$, 
    \begin{equation}
         \Pr{[\SignsTest{2k} \text{ rejects}]} \geq \min\set{ \frac{\eps_f}{2}, \frac{1}{10}}.
    \end{equation}
\end{lemma}

\begin{proof}

\Cref{lem:large distance soundness guarantee} immediately follows from the following lemma.

\begin{lemma}\label{lem:2mu-closeness}
\label{lem: low rejection prob implies f is close to hom}
    Let $\mu :=\Pr[\text{$\SignsTest{2k}$ rejects  
    $f$}].$ 
            If $\mu < \frac 1 {10}$, then $f$ is $2\mu$-close to a homomorphism.
\end{lemma}

Before proving \Cref{lem:2mu-closeness}, we explain why it implies \Cref{lem:large distance soundness guarantee}.
By \Cref{lem: low rejection prob implies f is close to hom},  if $\mu < 1/10$ then $f$ is $2\mu$-close to some homomorphism. Since $f$ has distance $\eps_f$ to $\mathsf{HOM}(G,H)$, we get $2\mu \geq \eps_f \implies \mu \geq \frac{\eps_f}{2}$. Thus, $\mu \geq \min \set{\frac{\eps_f}{2}, \frac{1}{10}}$. The rest of the proof is dedicated to proving \Cref{lem: low rejection prob implies f is close to hom}.

\paragraph{First observations.}
In the test, the distribution of the element $a$, as specified in \Cref{alg: fixed signs test}, Line~\ref{step:query}, is uniform over all elements of $G$, since for each realization
of the signs $\sigma_1, \dots , \sigma_{2k}$ and $x_1,\dots,x_{2k-1}$, there is a single choice of $x_{2k}$ leading to each $a$.
Define $\fixset{a}$ to be the set of realizations for $\sigma$'s and $x$'s leading to $a$,
\[
    \textstyle{\fixset{a} := \big\{(\sigma_1,\dots,\sigma_{2k}, x_1,\dots,x_{2k}) : \sum_{i \in [2k]} \sigma_i x_i = a\big\}} .
\]
We use $\tuple{\sigma}{x}$ as a short hand for $
(\sigma_1,\dots,\sigma_{2k}, x_1,\dots,x_{2k})$. 
A random tuple from $\fixset{a}$ is $k$-wise  independent in the following sense.
\begin{claim}
    \label{claim:half_the_elements_are_random}
    For all $a \in G$, all sets of indices $S \subseteq [2k]$ of size $\card{S} = k$, let  $(\overline{\sigma},\overline{x})$ be a uniformly random tuple from $\fixset{a}$. Then 
    \begin{enumerate}
        \item The vector $(\sigma_i, x_i)_{i\in S}$ is a uniformly distributed vector of $k$ signs and $k$ elements from $G$.
        \item The partial sum $\sum_{i\in S} \sigma_i x_i$ is a uniformly distributed element in $G$.
    \end{enumerate}    
\end{claim}
\begin{proof}
    {\bf Part 1.} Consider a tuple $(\sigma_i, x_i)_{i\in[2k]}$, not necessarily in $\fixset{a}$, and fix a partial assignment that gives values only to $(\sigma_i, x_i)_{i\in S}$. Now choose an index $j\in [2k]\setminus S$ and let $(\sigma_j, x_j)$ be chosen uniformly at random. Fill up the remaining indices, $(\sigma_i, x_i)$ for $i\in [2k] \setminus (S \cup \{j\})$ with arbitrary values. Since $(\sigma_j, x_j)$ is chosen uniformly at random, the sum $\sum_{i=1}^{j} \sigma_i x_i$ is a uniformly random element in $G$. Also, for any group $G$, the function $x \to x + g$ is a bijection from $G$ to $G$ for all $g\in G$ (the same holds true for $g+x, x-g, -g+x$). Thus, the rest of additions and subtractions (for $i = j+1,\dots, 2k$) are all bijections and therefore the entire sum $\sum_{i \in [2k]} \sigma_i x_i$ is also a uniformly random element of $G$. This means that exactly $1/\card{G}$ of the tuples with a particular partial assignment to $(\sigma_i, x_i)_{i\in S}$ are in fact in $\fixset{a}$. The same argument applies for any partial assignment of $(\sigma_i, x_i)_{i\in S}$, so each such partial assignment produces the exact same number of tuples in $\fixset{a}$.

    {\bf Part 2.} Consider the set $F_z$ which consists of all assignments to $(\sigma_i,x_i)_{i\in S}$ for which  $\sum_{i\in S} \sigma_i x_i = z$. Using similar arguments as in the proof of part 1, each set $F_z$ has the same size, and the same number of ways to complete the assignment to a tuple in $\fixset{a}$. Therefore, the number of tuples in $\fixset{a}$ for which $\sum_{i\in S} \sigma_i x_i = z$ is the same for each $z$.
\end{proof}

We consider tuples whose signed sums over either their left or right halves equal a fixed value.
Define
\begin{talign*}
     \leftside{z} := \big\{\tuple{\sigma}{x} : \sum_{i\in[k]} \sigma_i x_i = z \big\}\,, && \rightside{z} := \big\{\tuple{\sigma}{x} : \sum_{i \in [k]} \sigma_{k+i} x_{k+i} = z\big\} \,.
\end{talign*}

\begin{corollary}
    \label{cor:alternative_draw}
    Consider the sets $S_L = [k]$ and $S_R = [2k] \setminus [k]$.
    The following process produces a uniformly random tuple in $\fixset{a}$: 
    \begin{enumerate}
        \item Draw a uniformly random element $z\in G$.
        
        \item Draw a uniformly random tuple $(\overline{\sigma}^L, \overline{x}^L)$ from $\leftside{a - z}$ and take its ``left'' part, $(\overline{\sigma}^L, \overline{x}^L)_{i \in S_{L}}$ 
                
        \item Draw a uniformly random tuple $(\overline{\sigma}^R, \overline{x}^R)$ from $\rightside{z}$ and take its ``right'' part, $(\overline{\sigma}^R, \overline{x}^R)_{i \in S_R}$
                
        \item Return the tuple $(\overline{\sigma}, \overline{x})$ where $(\overline{\sigma}, \overline{x})_{i \in S_L} = (\overline{\sigma}^L, \overline{x}^L)_{i \in S_{L}}$ and $(\overline{\sigma}, \overline{x})_{i \in S_R} = (\overline{\sigma}^R, \overline{x}^R)_{i \in S_R}$ 
    \end{enumerate} 
\end{corollary}

\begin{proof}
    
        By the second part of \Cref{claim:half_the_elements_are_random}, the number of tuples $(\sigma_1,\dots, \sigma_{2k}, x_1, \dots, x_{2k}) \in \fixset{a}$ such that their right part is in $\rightside{z}$ (that is, $\sum_{i \in [k]} \sigma_{k+i} x_{k+i} = z$) is the same for each $z$.
    The same holds for the $\leftside{a - z}$.
    Since each tuple in $\fixset{a}$ has a unique such pair $(a - z, z)$, the process gives an equal probability for each tuple in $\fixset{a}$, concluding the proof. 
\end{proof}

\begin{definition}[Vote of $(\overline{\sigma}, \overline{x})$ for the value of $a$]
\label{def: vote-for-corrrector}
For all $a \in G$ and $(\overline{\sigma}, \overline{x}) \in \fixset{a}$, define $\bigoplus_{i \in [2k]} \sigma_i f(x_i)$
as the \emph{vote of $(\overline{\sigma}, \overline{x})$ for the value of $a$}.    
\end{definition}
\paragraph{Corrector.}

Define $g$ to be the local corrector for $f$, that is, the value that
has the most votes as defined in \Cref{def: vote-for-corrrector}:
\[
    \textstyle{g(a) := \text{argmax}_{h\in H} \bracket{\Pr_{(\overline{\sigma},\overline{x})\in\fixset{a}}{\bigoplus_{i \in [2k]} \sigma_i f(x_i) = h}} \text{ for all } a \in G}.
\]
Define the following probabilities, all relating to how close $f$ is to being homomorphic:
\begin{talign*}
    \eta_a:= \Pr_{(\overline{\sigma},\overline{x})\in\fixset{a}}{\big[ g(a) \neq \bigoplus_{i \in [2k]} \sigma_i f(x_i) \big]}
    && \eta:= max_a \eta_a
    && \delta:= \Pr_{x\in G}{[f(x) \neq g(x)]}
\end{talign*}
We follow the standard self-correction paradigm introduced by \cite{BLR93, Ben-orCLR08}, by showing that if the test fails with a small probability (e.g., $\mu < 1/10$), then: (1) $\eta$ is bounded; (2) $\delta \leq 2\mu$; and (3) $g$ is a homomorphism.
The last part (3) is evidently the harder one, although similar simpler arguments are used to prove (1). Altogether, they show that if the test fails with a small probability, then $f$ is close to being a homomorphism (and in particular, the homomorphism $g$). By contrapositive, this proves a soundness guarantee for the tester.

We start by formally stating and proving the first assertion.
\begin{lemma}
    \label{lem:corrector_bounding_eta}
     $\eta \leq 2\mu$.
\end{lemma}

\begin{proof}
    It suffices to show that $\eta_a \leq 2\mu$ for each $a \in G$.
    Fix some $a$ and consider two tuples  $(\overline{\sigma}, \overline{x}), (\overline{\tau}, \overline{y})$ chosen independently and uniformly at random from $\fixset{a}$. We show that the probability of the event that the votes of $(\overline{\sigma}, \overline{x})$ and $(\overline{\tau}, \overline{y})$ for the value of $a$  are equal is close to 1, which implies the most common vote among all tuples in $\fixset{a}$ occurs with probability close to $1$.
    To check whether they give     the same vote     we manipulate both summations as follows: apply from the left side $-\tau_k f(y_k) \dots - \tau_1 f(y_1)$ to cancel out the first $k$ summands from the tuple $(\overline{\tau},\overline{y})$, and similarly from the right side for the last $k$ summands of the other tuple. 
    We get the following equivalence:
    \[
        \textstyle{
        \underbrace{\bigoplus_{i \in [2k]} \sigma_i f(x_i)}_{S_{\sigma,x}}
        = 
        \underbrace{\bigoplus_{i \in [2k]} \tau_i f(y_i)}_{S_{\tau,y}}
        \iff 
        \underbrace{\bigoplus_{i=0}^{k-1}
        -\tau_{k-i} f(y_{k-i}) \oplus
        \bigoplus_{i\in[k]} \sigma_i f(x_i)}_{S_1} 
        = 
        \underbrace{\bigoplus_{i=k+1}^{2k} \tau_i f(y_i) \oplus \bigoplus_{i=0}^{k-1} -\sigma_{2k-i} f(x_{2k-i})}_{S_2}.
        }
    \] 
    Importantly, each of the new sums $S_1, S_2$
    contains exactly $k$ summands respectively from $(\overline{\sigma}, \overline{x})$ and $(\overline{\tau}, \overline{y})$, but with different signs. By \autoref{claim:half_the_elements_are_random} and independence between $(\overline{\sigma}, \overline{x})$ and $(\overline{\tau}, \overline{y})$, the tuple of $2k$ signs and $2k$ elements used in $S_1$ is uniformly random (and similarly for $S_2$).
     Since the original tuples are taken from $\fixset{a}$,
     \[
        \textstyle{\sum_{i\in[2k]} \sigma_i x_i 
        = a
        = \sum_{i\in[2k]} \tau_i y_i}.
    \]
         By the same reasoning     as before but with $x_i, y_i$ instead of $f(x_i), f(y_i)$
    \[
        \textstyle{\sum_{i=0}^{k-1}
        -\tau_{k-i} y_{k-i} +
        \sum_{i\in[k]} \sigma_i x_i
        = \sum_{i=k+1}^{2k} \tau_i y_i + \sum_{i=0}^{k-1} -\sigma_{2k-i} x_{2k-i}.}
    \]
    Thus, the two tuples defining the sums or votes $S_1, S_2$ are both in $\fixset{a'}$, for the same uniformly random element ${a'}\in G$. The event 
    $\set{S_1 \neq S_2}$ implies that at least one of the votes did not agree with $f({a'})$. 
    \[
        \textstyle{
        \Pr{[S_1 \neq S_2]}
        \leq \Pr{[S_1 \neq f({a'}) \vee S_2 \neq f({a'})]}
        \leq \Pr{[S_1 \neq f({a'})]} + \Pr{[S_2 \neq f({a'})]} = 2\mu ,}
    \]
    where the second inequality is by a union bound, and the  last equality holds     since     both the probabilities     are equal to the probability of the check in Line~\ref{step:signs-condition} of \Cref{alg: random signs test}     holds, thus simulating our test. Indeed, $S_i \in \fixset{a'}$ for $i=1,2$, where ${a'}$ is a uniformly random element of $G$. We denote by $p_h$ the probability that a correction using a random tuple $(\overline{\sigma},\overline{x})\in\fixset{a}$ produces $h\in H$, then:
    \[
        2\mu \geq \Pr{[S_1 \neq S_2]} = \Pr{[S_{\sigma,x} \neq S_{\tau,y}]}
        = \sum_{h\in H} p_h ( 1-p_h) 
        = \sum_{h\in H} (p_h - p_h^2)
        = 1 - \sum_{h\in H} p_h^2 .
    \]
        By definition, $g(a) = \argmax_{h\in H} p_h$ is the most common vote or correction outcome (ties broken arbitrarily). We have
    \[
        1 - \eta_a = p_{g(a)}
        = p_{g(a)} \cdot \sum_{h\in H} p_h 
        = \sum_{h\in H} p_h p_{g(a)}
        \geq \sum_{h\in H} p_h^2 \geq 1 - 2\mu ,
    \]
which concludes the proof.
\end{proof}

The next simple lemma shows the second easy assertion as a corollary of the first.
\begin{lemma}
    \label{lem:corrector_bounding_delta}
    If $\mu < 1/10$, then     $\delta \leq 2\mu$.
\end{lemma}

\begin{proof}
    For each element $a \in G$ we have high agreement for corrections, and in particular $\eta_a \leq 1/2$. Denote the set where $f$ and $g$ disagree by $\Delta:= \set{x \in G: f(x) \neq g(x)}$, noting that $\delta = \card{\Delta}/\card{G}$. Then,
    \[
        \mu 
        \ =\  \Pr_{\substack{a \in G\\{(\sigma,x)\in\fixset{a}}}}{\big[ f(a) \neq  \bigoplus_{i\in[2k]} \sigma_i f(x_i) \big]}
        \ =\  \frac{1}{\card{G}}\sum_{a \in G} \Pr_{(\sigma,x)\in\fixset{a}}{\big[ f(a) \neq  \bigoplus_{i\in[2k]} \sigma_i f(x_i) \big]} .
    \]
    The latest expression is only reduced if we take a subset of (positive) summands, only those from $\Delta \subseteq G$. It further reduces if we replace each event of the sum differing from $f(a)$ with the subevent that it specifically equals $g(a)$ (this is indeed a subevent for any $a \in \Delta$). We get
    \[
        \mu
        \geq \frac{1}{\card{G}}\sum_{a \in \Delta} \Pr_{(\sigma,x)\in\fixset{a}}{\big[ g(a) =  \bigoplus_{i\in[2k]} \sigma_i f(x_i) \big]} 
        = \frac{1}{\card{G}}\sum_{a \in \Delta} (1 - \eta_a)
        \geq \frac{1}{\card{G}} \cdot \card{\Delta} \cdot \frac{1}{2}
        = \frac{\delta}{2} ,
    \]
    where the last inequality uses $1-\eta_a \geq 1/2$ for any $a \in \Delta$. This concludes the proof.
\end{proof}

The last assertion uses the probabilistic method and the accuracy of corrections shown above, to prove that $g$ is a homomorphism.

\begin{lemma}
    \label{lem:corrector_g_homomorphisms}
    If $\mu < 1/10$, then $g$ is a homomorphism.    
\end{lemma}

\begin{proof}
    Given $a, {b} \in G$, we need to show that $g(a) \oplus g({b}) = g(a + b)$.
    Consider the following process:
    \begin{itemize}
        \item 
        Draw a uniformly random $r \in G$.         
        \item 
        Draw a uniformly random tuple $\tuple{\sigma}{x}\in \leftside{a - r} \bigcap \rightside{r}.
        $
        \item 
        Let $CL$ be the set ``cancelling left'' whose left side cancels exactly the right side of the previous tuple, defined as, $CL: = \set{\tuple{\tau}{y}: \forall i\in[k]. (\tau_i, y_i) = (-\sigma_{2k+1-i},x_{2k+1-i})}$, and choose uniformly at random a tuple $\tuple{\tau}{y}$ from $CL \bigcap \rightside{r + {b}}$.
    \end{itemize} 
    Applying \Cref{cor:alternative_draw} three times, note that $\tuple{\sigma}{x}$ is a uniformly random tuple in $\fixset{a}$, the tuple $\tuple{\tau}{y}$ is a uniformly random tuple in $\fixset{b}$. Furthermore, we define the tuple $\tuple{\phi}{z}$ that has the left side taken from $\tuple{\sigma}{x}$ and the right side taken from $\tuple{\tau}{y}$, that is $(\phi_i,z_i) = (\sigma_i,x_i)$ for $i\in[k]$ and $(\phi_i, z_i) = (\tau_i, y_i)$ for $i\in[2k]\setminus[k]$. Applying \Cref{cor:alternative_draw} again, we get that $\tuple{\phi}{z}$ is a unifromly random tuple from $\fixset{a+b}$. 
    Hencet, each of the following three equalities holds with probability at least $1-\eta$:
        \begin{align*}
        g(a) = \bigoplus_{i\in[2k]} \sigma_i f(x_i) &&
        g(b) = \bigoplus_{i\in[2k]} \tau_i f(y_i) &&
        g(a + b) = \bigoplus_{i\in[2k]} \phi_i f(z_i)
    \end{align*} 
    Next, recall we specifically chose the left side of $\tuple{\tau}{y}$ to cancel out the right side of $\tuple{\sigma}{x}$. This cancellation also happens     when applying the function $f$, since pairs of summands cancel each other from inside out. The following equality thus holds with probability $1$:
    \[
        \bigoplus_{i=k+1}^{2k} \sigma_i f(x_i) \oplus \bigoplus_{i\in[k]} \tau_i f(y_i) 
        = \sigma_{k+1}f(x_{k+1}) \oplus \dots \oplus \sigma_{2k}f(x_{2k}) \oplus \tau_1 f(y_1) \oplus \dots \oplus \tau_k f(y_k)
        = e_H .
    \]
    Over the random choice of tuples, the event that all $4$ equalities above hold has probability at least $1-3\eta \geq 1-6\mu > 0$ (by applying \Cref{lem:corrector_bounding_eta}, a union bound, and the premise $\mu < 1/6$).
    This event implies that $g(a) \oplus g(b) = g(a+b)$, which also occurs with positive probability. However, the last equality does not depend on the probability space, and therefore it must hold with probability $1$.
\end{proof}
This completes the proof of \Cref{lem:large distance soundness guarantee}.
\end{proof}

\subsection{Completing the soundness analysis}\label{sec:completing-the-soundness-analysis}

\begin{proof}[Proof of \Cref{thm:general-soundness-based-on-eps}]
Since $f$ is $\eps$-far from being a homomorphism, $\eps_f \geq \eps$.
   When $\eps_f \geq 1/8$, then by \Cref{lem:large distance soundness guarantee},
    \[\Pr{[\SignsTest{2k} \text{ rejects}]} \geq 
    \min\set{ \frac{\eps_f}{2}, \frac{1}{10}} \geq \frac{1}{16}.\]
    On the other hand, if $\eps_f < 1/8$, then by \Cref{lem:soundeness-signs-test-small-eps-general-H}, 
    \[\Pr{[\SignsTest{2k} \text{ rejects}]} \geq \min\set{\frac{(2k-3)\cdot \eps_f}{3}, \frac{1}{6}} \geq \min\set{\frac{(2k-3)\cdot \eps}{3}, \frac{1}{6}}.\]
    Combining the two cases, we get
    $$\Pr{[\SignsTest{2k} \text{ rejects}]} \geq \min\set{\frac{(2k-3)\cdot \eps}{3}, \frac{1}{6}, \frac{1}{16}} \geq \min\set{\frac{(2k-3)\cdot \eps}{3}, \frac{1}{16}}.\qedhere   
    $$
\end{proof}

\section{General framework for online manipulation resiliency}\label{sec:general_framework}
The previous section provided a sound basic test for homomorphism. To make this test resilient to online manipulations, we now introduce the general framework used to build our final testers. This framework, based on the concept of an 'unpredictable tester' and an amplification lemma, provides a generic recipe for converting a basic test with certain unpredictability properties into one that is robust against online adversaries. We will then apply this framework in the following sections to prove our main results.

\subsection{Unpredictable tester}\label{sec:unpredictable-tester}
We define flatness of a distribution, which closely relate to its ``min-entropy":\footnote{If a distribution is $\alpha$-flat, then by definition it has min-entropy at least $\log(1/\alpha)$.}

\begin{definition}[$\alpha$-flat]
A distribution $\cP$ over a finite domain $D$ is $\alpha$-flat if 
$$\max_{z \in D} \Pr_{x \sim \cP}[x = z] \leq \alpha.$$

\end{definition}

The following definition, loosely speaking, requires that other than probability $\beta$ of ``losing its unpredictability", a $q$-query test maintains``aggregated flatness" of $\alpha$. In the lemma that follows we show such a tester is resilient to online manipulations.

\begin{definition}[$(q, \alpha, \beta)$-Unpredictable Test]
\label{def:unpredictable-test}
A test making $q$ sequential queries $x_1, \dots, x_q$ to a function $f:D \to R$ is \textbf{$(q, \alpha, \beta)$-unpredictable} if there exist non-negative flatness parameters
$(\alpha_1, \dots, \alpha_q)$
and error terms $(\beta_1, \dots, \beta_q)$ satisfying two conditions:

\begin{enumerate}
    \item \textbf{Sequential Unpredictability:} For each $i \in [q]$, let $\cD_i$ be the probability distribution for the $i$-th query $x_i$, which may depend on all the previous queries $(x_1, \dots, x_{i-1})$.
    Then for each $i \in [q]$, with probability at least $1-\beta_i$ (taken over the random outcomes of the previous $i-1$ queries), $\cD_i$ is $\alpha_i$-flat i.e.,
    $$ \Pr_{x_1, \dots, x_{i-1}} \left[ \mathcal{D}_i \text{ is } \alpha_i\text{-flat} \right] \ge 1 - \beta_i $$

    \item \textbf{Aggregate Parameters:} The sum of flatness parameters and error terms is bounded:
        $$ \sum_{i \in [q]} \alpha_i \leq
        \alpha \quad \text{and} \quad \sum_{i \in [q]} \beta_i \leq
        \beta $$
\end{enumerate}

\end{definition}

\begin{lemma}[Probability of seeing a manipulation]
\label{lem: prob-of-manipulation}
Let $\cA$ be an algorithm that runs a $(q, \alpha, \beta)$-unpredictable test for $r$ independent iterations against any $t$-online budget-managing adversary. Then the probability that a manipulation is encountered in any single iteration, $P_{manipulation}$ is bounded by:
$$ P_{manipulation} \le \alpha qrt + \beta. $$
where $qrt$ is the maximum number of manipulations that the adversary can perform.
\end{lemma}

\begin{proof}
Fix any adversarial strategy. We solely rely on the bound on the number of manipulated entries and  the unpredictability of the test.
Since $\cA$ runs a $(q, \alpha, \beta)$-unpredictable test, there exist flatness parameters $(\alpha_1, \dots, \alpha_q)$ and error terms $(\beta_1, \dots, \beta_q)$ satisfying the conditions in \Cref{def:unpredictable-test}. 

Fix any iteration of the loop in algorithm $\cA$. Let $E_i$ be the event that the $i$-th query of this iteration, $x_i$, encounters a manipulation. 

We now bound the probability $\Pr[E_i]$ for a single $i$. Let $B_i$ be the bad event that the distribution $\cD_i$ is not $\alpha_i$-flat. 
Then $\Pr[B_i] \leq \beta_i.$

Suppose the bad event did not happen, i.e., the distribution $\cD_i$ is $\alpha_i$-flat.
Let $\cM_i$ be the set of manipulated entries when query $i$ is made, and note $\card{\cM_i} \leq qrt$.
Thus, the probability that $\cA$ encounters a manipulation while querying $x_i$ conditioned on $\overline{B_i}$ is 
\[
    \Pr[E_i|\overline{B_i}] 
    = \sum_{z \in \cM_i} \Pr_{x_i \sim \cD_i}{[x_i = z]}
    \leq \sum_{z \in \cM_i} \alpha_i  \leq \alpha_i qrt.
\]

The probability of seeing a manipulation while querying $x_i$ is thus bounded by
\begin{align*}
    \Pr[E_i] &= \Pr[E_i \mid \overline{B_i}]\Pr[\overline{B_i}] + \Pr[E_i \mid B_i]\Pr[B_i] \\
    &\leq \Pr[E_i \mid \overline{B_i}] + \Pr[B_i] \\
    &\leq \alpha_i qrt + \beta_i.
\end{align*}
By a union bound,
the probability of seeing a manipulation in this iteration is 
$$P_{erase} \le \sum_{i=1}^q \Pr[E_i] = \sum_{i=1}^q \big(\alpha_i qrt + \beta_i \big) \leq \alpha qrt + \beta.\qedhere$$
\end{proof}

\subsection{Generic Amplification Lemma}

In this section, we prove a generic amplification lemma. This result provides a recipe for converting a base test into a full, online manipulation-resilient tester, provided the base test has sufficient unpredictability. Specifically, we show that if a $(q, \alpha, \beta)$-unpredictable test satisfies a condition linking its soundness $p_w(\eps)$ to its aggregate unpredictability parameters $(\alpha, \beta)$ and the adversary's power ($t$), then running the test $O(1/p_w(\eps))$ times yields a new tester with a constant soundness guarantee against the adversary.

\begin{lemma}[Generic Amplification for Online Resilience]
\label{lem: generic amplification}
Let $\cP$ be a property (a set of functions $f:D \to R$). Let $\cT$ be a $(q, \alpha, \beta)$-unpredictable test for $\eps$-testing property $\cP$ with the following guarantees:
\begin{itemize}
    \item \textbf{Completeness:} If $f \in \cP$, the test $\cT$ always accepts.
    \item \textbf{Soundness:} If $f$ is $\eps$-far from $\cP$, $\cT$ finds a witness (i.e., rejects) with probability at least $p_w(\eps)$.
\end{itemize}
If the following condition is true,
\begin{equation}
\label{eq: amplification-condition}
     \alpha qt \cdot \Big \lceil \frac{3}{p_w(\eps)} \Big \rceil + \beta \le \frac{p_w(\eps)}{12} 
\end{equation}
then there is an online manipulation-resilient $\eps$-tester $\cA$ for the property $\cP$ which uses $ q \cdot \lceil 3/p_w(\eps) \rceil $ queries. 

Moreover $\cA$ has the structure that it runs $\lceil 3/p_w(\eps) \rceil$ iterations of the $(q, \alpha, \beta)$-unpredictable proximity oblivious test where the queries chosen in each iteration are independent of all queries chosen in previous iterations. If the manipulations were instead erasures, then $\cA$ has 1-sided error.
\end{lemma}

\begin{proof}
Let $r = \lceil 3/p_w(\eps) \rceil$  be the number of iterations of the online-resilient tester $\cA$.

\paragraph{Completeness:}
The tester $\cA$ accepts if all $r$ iterations accept. Thus $\cA$ can fail only if one of the iterations rejects. Since the base test $\cT$ has perfect completeness in the standard property testing model i.e., it always accepts functions in $\cP$ when there are no online manipulations, the rejection in an iteration can only happen because of a manipulation. By \Cref{lem: prob-of-manipulation} and condition \eqref{eq: amplification-condition} in the lemma statement, the probability of manipulation in any single iteration is at most

$$\alpha qrt + \beta = \alpha qt \Big \lceil \frac{3}{p_w(\eps)} \Big \rceil + \beta \leq \frac{p_w(\eps)}{12}.$$

By a union bound over all the $r = \Big \lceil \frac{3}{p_w(\eps)} \Big \rceil$ iterations, the probability that $\cA$ sees a manipulation and errors is at most 
$$\frac{p_w(\eps)}{12} \cdot \Big \lceil \frac{3}{p_w(\eps)} \Big \rceil \leq 1/3.$$

Note that if the manipulations were instead erasures, $\cA$ will always accept $f$ with probability 1 because $\cA$ can be designed to accept in any iteration where it sees an erasure.

\paragraph{Soundness:}
Assume $f$ is $\eps$-far from $\mathcal{P}$. The amplified tester $\cA$ fails only if all $r$ independent iterations fail to find a witness. Fix an iteration $i$. This iteration fails if one of the following two events happen, the base test $\cT$ fails to find a witness (denoted by event $\overline{W_i}$) or a manipulation occurs during its execution (denoted by event $M_i$). By a union bound, 
$$ \Pr[\text{Iteration } i \text{ fails}] = \Pr[\overline{W_i} \cup M_i] \leq \Pr[\overline{W_i}] + \Pr[M_i] $$

The probability that the base test $\cT$ finds a witness is at least $p_w(\eps)$. Thus $\Pr[W_i] \geq p_w(\eps)$ which implies  $\Pr[\overline{W_i}] \leq 1 - p_w(\eps)$. By \Cref{lem: prob-of-manipulation}, the probability that iteration $i$ sees a manipulation, $$\Pr[M_i] \leq \alpha qrt + \beta = \alpha qt \Big \lceil \frac{3}{p_w(\eps)} \Big \rceil + \beta.$$ 
Using condition \eqref{eq: amplification-condition} from the lemma statement, we get $\Pr[M_i] \leq \frac{p_w(\eps)}{12}$. Thus,

$$ \Pr[\text{Iteration } i \text{ fails}] \leq \Pr[\overline{W_i}] + \Pr[M_i] \leq (1 - p_w(\eps)) + \frac{p_w(\eps)}{12} = 1 - \frac{11 p_w(\eps)}{12}.$$
Let $Q_i$ be the random variable representing the set of queries made in iteration $i$. We know that for all choices of sets of queries for the previous $i-1$ iterations $z_1, \dots , z_{i-1}$, we have,

$$ \Pr_{Q_i}[\text{Iteration } i \text{ fails} | Q_1 = z_1, \dots , Q_{i-1} = z_{i-1}] \leq 1 - \frac{11 p_w(\eps)}{12}.$$ Thus, the expected probability of iteration $i$ failing over all the random choices of sets of queries made in all the previous $i-1$ iterations is
$$ \Pr_{Q_i}[\text{Iteration } i \text{ fails}] = \mathbb{E}_{z_1, \dots , z_{i-1}}\Big[ \Pr_{Q_i}[{\text{Iteration } i \text{ fails}} | Q_1 = z_1, \dots Q_{i-1} = z_{i-1}] \Big] \leq 1 - \frac{11p_w(\eps)}{12}.$$
Now, using the independence of the set of queries $Q_i$ for all $i \in [r]$,

$$ \Pr[\cA \text{ fails}] = \prod_{i \in [r]} \Pr_{Q_i}[\text{Iteration } i \text{ fails}] \leq \prod_{i \in [r]} 1 - \frac{11p_w(\eps)}{12} = \Big( 1 - \frac{11 p_w(\eps)}{12}\Big)^r.$$
Substituting $r = \lceil 3/p_w(\eps) \rceil$, we know $r \ge 3/p_w(\eps)$, thus
$$ \Pr[\cA \text{ fails}] \le \left( 1 - \frac{11 p_w(\eps)}{12} \right)^{3/p_w(\eps)} \leq e^{\frac{-11 p_w(\eps)}{12} \cdot \frac{3}{p_w(\eps)}} = \left( \frac{1}{e} \right)^{11/4} \leq \frac{1}{3},$$
where the second inequality holds since $1-x \leq e^{-x}$ for all $x$. Thus, the failure probability of $\cA$ is at most
1/3, and $\cA$ successfully finds a witness and rejects $f$ with probability at least 2/3.
\end{proof}

\section{Resilience to online manipulations for general groups}\label{sec:resilience}
In this section, we prove \Cref{thm:general_G_H} by lifting the linearity tester in \cite{AroraKM25} which is based on the manipulation-resilient tester of \cite{BenEliezerKMR24} and the sample-based tester of Goldreich and Ron \cite{GoldreichR16}. Define a parameter $m := 4 \lceil \log_2 t + 15/\eps \rceil + 12$. Depending on the value of $m$, we invoke one of the two different algorithms against the $t$-online-erasure adversary. When $2^m \leq \card{G}^{1/4}$, we use \Cref{alg: Online erasure resilent signs test}, a lift of the linearity tester of \cite{BenEliezerKMR24} to the general group setting. When $2^m > \card{G}^{1/4}$, we instead use the sample-based tester by \cite{GoldreichR16}, following the same case-based strategy as \cite{AroraKM25}. We analyze the two cases separately. Our main contribution here is in the analysis of the distribution of the element $y$ in \Cref{alg: unpredictable signs test} where we need to handle random signs along with random elements for any group $G$, whereas previous work did not have to worry about the random signs.

\paragraph{Case 1: If $2^m \leq \card{G}^{1/4}$.}In this case, we use \Cref{alg: Online erasure resilent signs test} which internally runs \Cref{alg: unpredictable signs test} a constant number of times.

\begin{algorithm}
\caption{Unpredictable Random Signs Test}
\label{alg: unpredictable signs test}
\begin{algorithmic}[1]
\Ensure $m \in \N$, $m$ is even.
\Require Query access to $f: G \to H$ where $(G,+)$ and $(H, \oplus)$ are finite groups. 
\State Draw $m$ independent and uniformly random elements $x_1, x_2, \dots, x_m$ from $G$.
\label{line:online-erasure-resilient-signs-test-line-3}
\State Query $f(x_1), f(x_2), \dots, f(x_m)$. \label{line:online-erasure-resilient-signs-test-line-4}
\State Draw independent and uniformly random signs $\sigma_j$ for all $j \in [m]$ from the set $\set{+,-}$.
\State Let $S$ be a uniformly random subset of $[m]$ of size $m/2$.
\State Query $f(y)$ where  $y \gets \sum_{j\in S} \sigma_j x_j$ \label{line:online-erasure-resilient-signs-test-line-7}
\State \textbf{Accept} if $\bigoplus_{j \in S} \sigma_j f(x_j) \neq f(y)$; otherwise \textbf{reject} \label{line:online-erasure-resilient-signs-test-line-8}
\end{algorithmic}
\end{algorithm}

\begin{algorithm}
\caption{Online Erasure-Resilient Random Signs Test}
\label{alg: Online erasure resilent signs test}
\begin{algorithmic}[1]
\Ensure $\eps \in (0,1)$, $t \in \N$
\Require Query access to $f: G \to H$ via $t$-erasure oracle, where $(G,+)$ and $(H, \oplus)$ are finite groups. 
\State Set $m = 4 \lceil \log_2 t + 15/\eps \rceil + 12$  
\label{line:online-erasure-resilient-signs-test-line-1}
\For{$i \in [48]$ 
} \label{line:online-erasure-resilient-signs-test-line-2}
\State Run \Cref{alg: unpredictable signs test} with parameter $m$ and independent random coins with query access to $f$.
\State \textbf{Reject} if \Cref{alg: unpredictable signs test} rejects.
\EndFor
\State \textbf{Accept}
\end{algorithmic}
\end{algorithm}

Firstly, \Cref{lem: unpredictable random signs completeness soundness} shows that \Cref{alg: unpredictable signs test} has the same completeness and soundness guarantees as $\SignsTest{2k}$ if $k = m/4$. 

\begin{lemma}
\label{lem: unpredictable random signs completeness soundness}

For all $m \in \N$ divisible by 4, the Unpredictable Random Signs Test stated in \Cref{alg: unpredictable signs test} with parameter $m$ has the same completeness and soundness guarantees as $\SignsTest{2k}$ where $k = m/4$.
\end{lemma}

\begin{proof}
Notice that \Cref{alg: unpredictable signs test} is essentially the same as \SignsTestInText stated in \Cref{alg: random signs test} with the some addition queries being made. The final check in line 5 of \Cref{alg: unpredictable signs test} uses $m/2$ independent and uniformly random elements and signs from $G$ and $\set{+,-}$ respectively which is the same check done in \Cref{alg: random signs test} since $2k = m/2$. As a result the completeness and soundness guarantees for \Cref{alg: unpredictable signs test} are the same as of $\SignsTest{2k}$ for $k = m/4$.
\end{proof}

Next we will show that  \Cref{alg: unpredictable signs test} apart from simulating the \SignsTestInText also has some unpredictability. To do that, we first prove \Cref{lem: Dx is flat with high prob} which shows that the distribution of the last query $y$ in \Cref{alg: unpredictable signs test} is flat with high probability.

\begin{lemma}
\label{lem: Dx is flat with high prob}
If $2^m \leq \card{G}^{1/4}$, then the distribution of the last query $y$ in \Cref{alg: unpredictable signs test} which depends on the $m$ initial queries $X$, denoted as $\cD^X$ is $1/{\binom{m}{m/2}}$-flat except with probability $1/2^m$ over the randomness of $X$.
\end{lemma}

\begin{proof}[Proof of \Cref{lem: Dx is flat with high prob}]

Once \Cref{alg: unpredictable signs test} chooses $X = \set{x_1, \dots , x_m}$ the distribution of $y$, $\cD^X$ is determined by the random choices of the signs $\overline{\sigma} = (\sigma_1, \dots, \sigma_m)$ and the subset $S \subseteq [m]$. Let $\cD^X_{\overline{\sigma}}$ be distribution conditioned on $X$ and the choice of signs $\overline{\sigma}$. Since the sequence of signs is chosen uniformly at random and independently from $S$, $\cD^X$ is a uniform mixture of all distributions $\cD^X_{\overline{\sigma}}$ for each choice of $\overline{\sigma}$. We can denote this mixture as follows.
    \[ \cD^X = \frac{1}{2^m} \sum_{\overline{\sigma} \in \set{+,-}^m} \cD^X_{\overline{\sigma}}.\]

\begin{definition}[Support of $\cD^X_{\overline{\sigma}}$]
       For all $X \in G^m$ and for all sequences of signs $\overline{\sigma} = (\sigma_1, \sigma_2, \dots, \sigma_m) \in \set{+,-}^m$, let $Y^X_{\overline{\sigma}}$ be the support of the distribution $\cD^X_{\overline{\sigma}}$, defined as
    $$Y^X_{\overline{\sigma}} = \big\{y \in G : \exists S \subseteq [m]\ , \card{S} = m/2 \text{ such that } \sum_{j \in S} \sigma_j x_j = y\big\}.$$
\end{definition}

By definition, the size of the support $Y^X_{\overline{\sigma}}$ can be at most ${\binom{m}{m/2}}$ as there can be at most ${\binom{m}{m/2}}$ distinct choices for $y$, one for each subset $S \subseteq [m]$ of size $m/2$. However, we show that the size of the support i.e., $\card{Y^X_{\overline{\sigma}}}$ is exactly equal to ${\binom{m}{m/2}}$ for all $\overline{\sigma} \in \set{+,-}^m$ simultaneously with high probability. 

Let $E$ be the event that the above statement is false, i.e., there exists a sequence of signs $\overline{\sigma}$ such that $\card{Y^X_{\overline{\sigma}}} < {\binom{m}{m/2}}$. The following claim bounds the probability of event $E$.

\begin{claim}
\label{clm: upper-bound-event-E}
     $\Pr_{X}{[E]} = \Pr_{X}{\sq{ \exists \overline{\sigma} \in \set{+,-}^m \text{ such that }\card{Y^X_{\overline{\sigma}}} < {\binom{m}{m/2}}}} \leq \frac{1}{2^m}$. 
\end{claim}
\begin{proof}[Proof of \Cref{clm: upper-bound-event-E}]
     For all subsets $S \subseteq [m]$, let $y_S = \sum_{j \in S} \sigma_j x_j$. The size of the support i.e., $\card{Y^X_{\overline{\sigma}}}$ can be  strictly less than ${\binom{m}{m/2}}$ only if some of the $y_S$ collide for different choices of $S$. We first show that for any $\overline{\sigma}$, the probability of such collisions is low and then show the above holds simultaneously for all sequences of signs.
    
    Let $S_1, S_2 \subseteq [m]$ (not necessarily of size $m/2$) such that $S_1 \neq S_2$. Without loss of generality,  assume $S_1 \setminus S_2$ is not empty and  select an index $i \in S_1 \setminus S_2$. Fix all values in $X$ apart from $x_i$. Then the value $y_{S_2}$ is fixed, but $y_{S_1}$ is uniformly random over $G$ as $x_i$ (or $-x_i$ depending on what $\sigma_i$ is) is uniformly random over $G$. Thus,
    \[\Pr_{X}{[y_{S_1} = y_{S_2}]} = \E{\Pr_{x_i}{[y_{S_1} = y_{S_2}]}} = \frac{1}{\card{G}}, \]
    where the expectation is over all $x_j \in X, j \neq i$. By a union bound over all choices of $S_1$ and $S_2$, 
    
    \[\Pr_{X}{[\exists S_1 \neq S_2: y_{S_1} = y_{S_2}]} \leq \frac{2^{2m}}{|G|} \implies \Pr_{X}{\sq{\card{Y^X_{\overline{\sigma}}} < {\binom{m}{m/2}}}} \leq \frac{2^{2m}}{|G|} \]
    By a union bound over all choices of signs $\overline{\sigma} \in \set{+,-}^m$, we get
    \[\Pr_{X}{\sq{ \exists \overline{\sigma} \in \set{+,-}^m \text{ such that }\card{Y^X_{\overline{\sigma}}} < {\binom{m}{m/2}}}} \leq \frac{2^{3m}}{\card{G}} \leq \frac{1}{2^m},\]
    where the last inequality is due to $2^m \leq \card{G}^{1/4} \implies 2^{4m} \leq \card{G} \implies 1/\card{G} \leq 1/2^{4m}$. 
\end{proof}

Now, we show that the distribution $\cD^X$ is $1/{\binom{m}{m/2}}$-flat with probability at most $1/2^m$.   Let $B$ be the event that $\Pr_{y \sim \cD^X}{[y = z]} > 1/{\binom{m}{m/2}}$ for some element $z \in G$. Assume  event $E$ does not happen, i.e., all the supports have size $\card{Y^X_{\overline{\sigma}}} = {\binom{m}{m/2}}$ for all $\overline{\sigma} \in \set{+,-}^m$. Since the set $S \subseteq [m]$ of size $m/2$ is chosen uniformly at random, the distribution $\cD^X_{\overline{\sigma}}$ is uniform over the set $Y^X_{\overline{\sigma}}$. Thus, 
    \[\Pr_{y \sim \cD^X_{\overline{\sigma}}}{\sq{y = z}} = \begin{cases}
        1/{\binom{m}{m/2}} & z \in Y^X_{\overline{\sigma}} \\
        0 & z \in G \setminus Y^X_{\overline{\sigma}},
    \end{cases}
    \]
     which implies that for all $z \in G$, 
    \[\Pr_{y \sim \cD^X}[y = z] = \frac{1}{2^m} \sum_{\overline{\sigma} \in \set{+,-}^m}  \Pr_{y \sim \cD^X_{\overline{\sigma}}}[y = z] \leq  \frac{1}{2^m} \sum_{\overline{\sigma} \in \set{+,-}^m} \frac{1}{{\binom{m}{m/2}}} = \frac{1}{{\binom{m}{m/2}}}. \]

    If event $E$ does not occur, then event $B$ cannot occur either. Thus, $E$ is a necessary condition for $B$, which implies $\Pr_{X}{[B]} \leq \Pr_{X}{[E]} \leq 1/2^m$. This completes the proof of \Cref{lem: Dx is flat with high prob}.
\end{proof}

Now, we are ready to show show that \Cref{alg: unpredictable signs test} is an unpredictable tester in \Cref{lem: unpredictable signs test unpredictability}. 

\begin{lemma}
\label{lem: unpredictable signs test unpredictability}
Fix a finite group $G$. If $2^m \leq \card{G}^{1/4}$, then \Cref{alg: unpredictable signs test} with parameter $m$ is 
$$\Big (m+1,\frac{m}{2^m} + \frac{1}{\binom{m}{m/2}}, \frac{1}{2^m} \Big )\text{-unpredictable.}$$
\end{lemma}
\begin{proof}
Firstly, note the \Cref{alg: unpredictable signs test} makes exactly $m+1$ queries. Consider all queries $x_i$ for $i \in [m]$ i.e., except the last one. From line 1 of \Cref{alg: unpredictable signs test}, we know that irrespective of the previous $i-1$ queries, $x_i$ is chosen independently and uniformly at random from $G$. Thus, the distribution of $x_i$ is the uniform distribution over $G$ which is $1/|G|$-flat with probability 1. Thus $\alpha_i = 1/|G|$ and $\beta_i = 0$ for all $i \in [m].$

Next we analyze what happens for the last query $y$. Let $\cD^X$ be the final distribution of $y$ conditioned on $X$. Using \Cref{lem: Dx is flat with high prob} we know that $\cD^X$ is $1/{\binom{m}{m/2}}$-flat except with probability $1/2^m$ over the randomness of $X$. Thus $\alpha_{m+1} = 1/{\binom{m}{m/2}}$ and $\beta_{m+1} = 1/2^m$.  Thus \Cref{alg: unpredictable signs test} is unpredictable with $$\sum_{ i \in [m+1]} \alpha_i = \frac{m}{|G|} + \frac{1}{{\binom{m}{m/2}}} \leq \frac{m}{2^m} + \frac{1}{{\binom{m}{m/2}}} = \alpha \text{ and } \sum_{ i \in [m+1]} \beta_i = \frac{1}{2^m} = \beta.$$
using the fact $2^m \leq |G|^{1/4}$.
\end{proof}

 Finally, we show that \Cref{alg: Online erasure resilent signs test} satisfies \Cref{thm:general_G_H} for case 1, when $2^m \leq \card{G}^{1/4}$. 
\begin{proof}[Proof of \Cref{thm:general_G_H} when $2^m \leq \card{G}^{1/4}$ ]

We will use \Cref{lem: generic amplification} to show that \Cref{alg: Online erasure resilent signs test} is an online erasure resilient tester. Let $m$ be as defined in line 1 of \Cref{alg: Online erasure resilent signs test}, i.e., $m = 4 \lceil \log_2 t + 15/\eps \rceil + 12$.

Let \Cref{alg: unpredictable signs test} be the base test needed in the statement of \Cref{lem: generic amplification}. Using \Cref{lem: unpredictable random signs completeness soundness}, we know that \Cref{alg: unpredictable signs test} with parameter $m$ has the same completeness and soundness guarantees as $\SignsTest{2k}$ where $k = m/4$. Thus, by \Cref{thm:general-soundness-based-on-eps}, \Cref{alg: unpredictable signs test} accepts all homomorphisms from $G$ to $H$ and rejects functions that are $\eps$-far from $\mathsf{HOM}(G,H)$ with probability at least $\min\set{\frac{(m/2-3)\cdot \eps}{3}, \frac{1}{16}} = \min\set{\frac{(m-6)\cdot \eps}{6}, \frac{1}{16}} \geq \frac 1{16}$ as 

\[ \frac{(m-6) \cdot \eps}{6}  = \frac{\big(4 \lceil \log_2 t + 15/\eps \rceil + 6 \big) \cdot \eps}{6} > \frac{(1/\eps) \cdot \eps}{6} = \frac 1 6 > \frac 1{16}.\]
Thus, the quantity $p_w(\eps) = 1/16$ as needed for using \Cref{lem: generic amplification}.  Using \Cref{lem: unpredictable signs test unpredictability}, we know that \Cref{alg: unpredictable signs test} is $(q, \alpha, \beta)$-unpredictable where $$q = m+1, \alpha = \frac{m}{2^m} + \frac{1}{{\binom{m}{m/2}}}, \beta = \frac{1}{2^m}.$$

To apply \Cref{lem: generic amplification}, we need the following condition (inequality \eqref{eq: amplification-condition} in the statemnt of \Cref{lem: generic amplification}) to be satisfied.
\[\alpha qt \cdot \Big \lceil \frac{3}{p_w(\eps)} \Big \rceil + \beta \le \frac{p_w(\eps)}{12} \]

Now we show that the above condition is satisfied by the values $\alpha, \beta, p_w(\eps)$ that we derived above.
\begin{align*}
        \alpha qt \cdot \Big \lceil \frac{3}{p_w(\eps)} \Big \rceil + \beta &= \Big( \frac{m}{2^m} + \frac{1}{{\binom{m}{m/2}}}\Big) \cdot (m+1)t \cdot 48 + \frac{1}{2^m}\\
        & \leq \frac{48t(m+1)m}{2^m} + \frac{48t(m+1)}{{\binom{m}{m/2}}} + \frac{1}{2^m} \\
        & \leq \frac{48t(m+1)m}{2^m} + \frac{48t(m+1)\sqrt{2m}}{2^m} + \frac{1}{2^m} \leq \frac{256 t(m+1)m}{2^m}.
    \end{align*}
    where we used ${\binom{m}{m/2}} \geq 2^m/\sqrt{2m}$ in the second inequality.  
    Using $m = 4 \lceil \log_2 t + 15/\eps \rceil + 12$, we get $ t \leq 2^{m/4 - 15/\eps - 3} \leq 2^{m/4 -18}$. Thus,

    \[\frac{256 t(m+1)m}{2^m} \leq \frac{256 (m+1)m}{2^{\frac{3m}{4} + 18}} = \frac{(m+1)m}{2^{\frac{3m}{4}} \cdot 1024} \leq \frac{3}{1024} < \frac{1}{12 \cdot 16} = \frac{p_w(\eps)}{12}.\]
    where the last inequality holds since $\frac{(m+1)m}{2^{\frac{3m}{4}}} \leq 3$ for all $m$.

Thus using \Cref{lem: generic amplification}, we conclude that \Cref{alg: Online erasure resilent signs test} that runs $ \lceil 3/p_w(\eps) \rceil = 48$ independent iterations of \Cref{alg: unpredictable signs test} with $m = 4 \lceil \log_2 t + 15/\eps \rceil + 12$ is an online manipulation resilient tester.
\end{proof}

\paragraph{Case 2: If $2^m > \card{G}^{1/4}$.} In this case, we use the sample-based tester by Goldreich and Ron \cite{GoldreichR16} stated in \Cref{alg: sample-based tester goldreich ron}.

\begin{definition}[Partial Sums of $X$]
    For all groups $G$ and all sequences of $m$ elements $X = \set{x_1, x_2, \dots, x_m}$ from $G$, let $\mathsf{Partial Sums}(X)$ be the set of all partial sums of $X$ i.e., all sums of the form $\sum_{i \in I} x_i$ for all $I \subseteq [m]$. 
\end{definition}

\begin{algorithm}
\caption{Sample-based homomorphism tester by Goldreich and Ron \cite{GoldreichR16}}
\label{alg: sample-based tester goldreich ron}
\begin{algorithmic}[1]
\Ensure $\eps\in(0,1)$
\Require Query access to a function $f: G \to H$ where $(G,+)$ and $(H, \oplus)$ are finite groups.
\State Sample a set $S$ of $O(\log \card{G})$ elements chosen independently and uniformly at random from $G$.
\State If $\mathsf{Partial Sums}(S) \neq G$; \textbf{accept}.
\State Query $f(x)$ for all $x \in S$. 
\State Let $h: G\to H$ be the unique homomorphism with $h(x)=f(x)$ for all $x\in S$; \textbf{reject} if no such $h$ exists.
\State Query $f$ on $O\big(\frac 1 \eps)$ uniform and independent elements from $G$.
\State \textbf{Reject} if $f(x)\neq h(x)$ for any sampled point $x$; otherwise, \textbf{accept}. \end{algorithmic}
\end{algorithm}

\begin{theorem}[Theorem 5.1 of \cite{GoldreichR16}]
\label{thm: goldreich ron sample-based testing}
For all finite groups $(G,+)$ and $(H,\oplus)$, all $\eps \in (0,1)$, \Cref{alg: sample-based tester goldreich ron} is a sample-based $\eps$-tester for group homomorphism for functions $f: G \to H$ that makes $O(1/\eps + \log \card{G})$ queries. 
\end{theorem}

Now we state a theorem by Arora, Kelman and Meir \cite{AroraKM25} about online-manipulation-resilience of sample-based testers.

\begin{theorem}[Theorem 2.1 of \cite{AroraKM25}]
\label{thm: sample-based-to-erasure-resilient}
    Let $\cT$ be a sample-based tester for a property $\mathcal{P}$ with input length $N$ and distance parameter $\varepsilon\in(0,1)$ that uses $q$ uniformly random samples and succeeds with probability $1 - \delta$. Then $\cT$ with $q$ queries succeeds with probability $1 - 2\delta$ in the presence of every budget-managing $t$-online-manipulation adversary for all $t \leq \delta \cdot N / q^2$.
\end{theorem}

Finally we show that \Cref{alg: sample-based tester goldreich ron} satisfies \Cref{thm:general_G_H} for case 2, when $2^m > \card{G}^{1/4}$. 

\begin{proof}[Proof of \Cref{thm:general_G_H} when $2^m > \card{G}^{1/4}$ ]
Using \Cref{thm: goldreich ron sample-based testing} we know \Cref{alg: sample-based tester goldreich ron} is a sample-based tester. By \Cref{thm: sample-based-to-erasure-resilient}, \Cref{alg: sample-based tester goldreich ron} also works in presence of a $t$-online-manipulation adversary for $t \leq c \cdot \min \set{\eps^2, 1/\log^2 \card{G}} \cdot \card{G}$ for some $c > 0$ and succeeds with probability at least 2/3. The query complexity of \Cref{alg: sample-based tester goldreich ron} is $O(1/\eps + \log \card{G}),$ which is at most $ O(1/\eps + m) = O(1/\eps + \log t)$. 
\end{proof}

\section{Group-specific testing}\label{sec:group-specif}
In this section, we present two approaches for improving testing bounds by leveraging the structure of specific groups $G$ and $H$. We present two approaches: one based on group-specific sample-based testing, and another that applies when $G$ and $H$ are finite fields.

\subsection{Group-specific sample-based homomorphism testing}\label{sec:group-specific-sample-based}
In this section, we prove \Cref{thm:group-specific-sample-based-erasure-resilient-result}, which gives a group-specific homomorphism tester  with complexity expressed in terms of the parameter $E(G)$, specified in \Cref{def:group-parameter-E}.
The tester that satisfies \Cref{thm:group-specific-sample-based-erasure-resilient-result} is \Cref{alg:generated subgroup test}. It samples elements of the group until it is likely to get a set that generates the entire group and then compares the homomorphism defined by the values of the input function $f$ on the generating set with the values of $f$ on another random sample of points.

\begin{definition}[Subgroup $\ab{S}$]
    The \emph{subgroup generated by a subset $S$} of a group $G$, denoted by $\ab{S}$, consists of all elements of $G$ that can be expressed as the finite product of elements of $S$ and their inverses.
\end{definition}

\begin{algorithm}
\caption{$\GeneratedSubgroupTest$}
\label{alg:generated subgroup test}
\begin{algorithmic}[1]
\Ensure $\eps\in(0,1)$
\Require Query access to a function $f: G \to H$ where $(G,+)$ and $(H, \oplus)$ are finite groups.
\State Set $m \gets E(G) +  9$.
\State Sample a set $S$ of $m$ elements chosen independently and uniformly at random from $G$.
\State If $\ab{S} \neq G$; \textbf{accept}.
\State Query $f(x)$ for all $x \in S$. 
\State Let $h: G\to H$ be the unique homomorphism with $h(x)=f(x)$ for all $x\in S$; \textbf{reject} if no such $h$ exists.
\State Query $f$ on $\frac 3 \eps$ uniform and independent elements from $G$.
\State \textbf{Reject} if $f(x)\neq h(x)$ for any sampled point $x$; otherwise, \textbf{accept}. \end{algorithmic}
\end{algorithm}

Observe $\GeneratedSubgroupTest$ (\Cref{alg:generated subgroup test}) is sample-based, i.e., it queries only independent and uniform random elements of $G$. To analyze the tester,
we first examine the probability that $\ab{S}=G$ for a set $S$ sampled from~$G$. The relevant group parameters were introduced by Menezes~\cite{menezes2013random} and Lubotzky~\cite{lubotzky2002expected}. We relate them to $E(G)$, the expected number of elements needed to sample a generator for the whole group, in \Cref{lem:bound-on-generating-G-whp}.
Then we use \Cref{thm: sample-based-to-erasure-resilient}
that shows that every sample-based tester is also online-erasure-resilient.

Finally, we complete the proof of \Cref{thm:group-specific-sample-based-erasure-resilient-result} by putting all ingredients together.

\begin{definition}[Group parameters $d^{\beta}(G)$ \cite{menezes2013random} and $\cM(G)$ \cite{lubotzky2002expected}]\label{def:d-beta-and-M-of-G}
    Let $G$ be a finite group.
    \begin{itemize}    
        \item For $\beta \in (0,1)$, let $d^{\beta}(G)$ be the number of independent, uniformly distributed elements of $G$ needed to generate $G$ with probability at least $1-\beta$.

        \item Let $\cM(G) := \sup_{r \geq 2} \frac{\log m_r(G)}{\log r}$ where $m_r(G)$ is the number of maximal subgroups of $G$ of index $r$. 
    \end{itemize}
\end{definition}

Menezes \cite[Proposition 9.2.7]{menezes2013random} showed a family of inequalities relating the group parameters, $d^{\beta}(G)$ and $\cM(G)$.

\begin{lemma}[Relating $d_\beta(G)$ and $\cM(G)$ \cite{menezes2013random}]
\label{lem:relate-dbeta-M}
    Let $\zeta(x)$ be the Reimann Zeta function. For each $w \in \R$ such that $\zeta(w) \leq 1+ \beta$, we have $d^{\beta}(G) \leq \cM(G) + w$. 
\end{lemma}

Lucchini and Moscatiello \cite[Theorem 1.1]{lucchini2020generation} showed a result relating the group parameters $\cM(G)$ and $E(G)$.

\begin{lemma}[Relating $E(G)$ and $\cM(G)$ \cite{lucchini2020generation}]
\label{lem:relate-E-M}
    Let $G$ be any finite group. Then,
    $\lceil \cM(G) \rceil \leq E(G) + 4$.
\end{lemma}

Next we use \Cref{lem:relate-dbeta-M} and \Cref{lem:relate-E-M} to relate $d^\beta(G)$ and $E(G)$. 

\begin{lemma}[Relating $d_\beta(G)$ and $E(G)$]
\label{lem:bound-on-generating-G-whp}
    For all $\beta \in (0,1)$ and for all finite groups $G$,  
    $$d\textstyle^{\beta}(G) \leq E(G) + 5 + \log_2 \big(1+ \frac{1}{\beta} \big).$$
\end{lemma}
\begin{proof} 
  Let $\zeta(x)$ be the Reimann Zeta function. Using \Cref{clm: upper bound on reimann zeta}, we know that for $w = 1 + \log_2 \big(1+ \frac{1}{\beta} \big)$, we have $\zeta(w) \leq 1+ \beta$. Applying the upper bound for $d^\beta(G)$ from \Cref{lem:relate-dbeta-M} using $w = 1 + \log_2 \big(1+ \frac{1}{\beta} \big)$, we get
  
    \[
        d^{\beta}(G)  \leq \cM(G) + 1 + \log_2 \big(1+ \frac{1}{\beta} \big)
        \leq E(G) + 5 + \log_2 \big(1+ \frac{1}{\beta} \big)\,,
    \]
    
where the last equality follows from \Cref{lem:relate-E-M} and the fact that $\cM(G) \leq \lceil \cM(G) \rceil$.
\end{proof}

Now we complete the proof of the main theorem from this section.
\begin{proof}[Proof of \Cref{thm:group-specific-sample-based-erasure-resilient-result}]
First, we analyze \Cref{alg:generated subgroup test} in the standard property testing model (without adversarial manipulations). The tester always accepts every homomorphism $f \in \mathsf{HOM}(G,H)$ because it only rejects if it finds an inconsistency. 

Now assume $f$ is $\eps$-far from $\mathsf{HOM}(G,H)$.  In this case, there are two bad events when \Cref{alg:generated subgroup test} incorrectly accepts: $B_1$ when $S$ does not generate $G$ and $B_2$ when $S$ generates $G$ but the algorithm still accepts $f$. By \Cref{lem:bound-on-generating-G-whp} with $\beta = \frac 1{12}$, one needs at most $E(G) + 5 + \log_2(13) \leq E(G) + 9$ samples from $G$ to generate $G$ with probability at least $\frac{11}{12}$. Since \Cref{alg:generated subgroup test} uses  $E(G) + 9$ samples, $\Pr{[B_1]} \leq \frac 1{12}$. Now suppose that $S$ generates $G$. If no unique homomorphism $h$ exists, the algorithm rejects. Otherwise, $f$ is $\eps$-far from $h$ as $f$ is $\eps$-far from $\mathsf{HOM}(G,H)$. The probability $f$ and $h$ agree on all $\frac 3\eps$ random checks is at most $(1-\eps)^{3/\eps} \leq e^{-3} < \frac 1{12}$. Thus, $\Pr[B_2] \leq \frac 1 {12}$. The total acceptance probability is at most $\frac 1 6$.

\paragraph{Query complexity.} The number of queries made by \Cref{alg:generated subgroup test} is $E(G) + 9 + \frac 3 \eps = O(\frac 1 \eps + E(G))$.

\paragraph{Manipulation Resilience.} We proved that \Cref{alg:generated subgroup test} is a sample-based $\eps$-tester for homomorphisms from $G \to H$ that succeeds with probability at least $\frac 5 6$. By
\Cref{thm: sample-based-to-erasure-resilient}, \Cref{alg:generated subgroup test} also works in presence of a $t$-online-manipulation adversary for $t \leq c \cdot \min \set{\eps^2, 1/E(G)^2} \cdot \card{G} $ and succeeds with probability at least $\frac 23$ for some constant $c> 0$.
\end{proof}

\subsection{Bounds for prime fields}\label{sec:prime-fields}
In this section, we prove \Cref{thm:prime-fields-p=q,thm:prime-fields-lower-bound}, which give a homomorphism tester for $G = \F_p^n$ and $H = \F_p^r$
and the corresponding lower bound $G = \F_p^n$ and $H = \F_p$.
The tester that satisfies \Cref{thm:prime-fields-p=q}
is \Cref{alg: Online erasure resilent coeffs test} which is based on the \CoeffsTestInText (\Cref{alg: random coeffs test}). This algorithm is similar to \SignsTestInText
(\Cref{alg: random signs test}) but uses coefficients from $[p-1]$ instead of signs.
For any vector $v \in \F_p^n$ and coefficient $\sigma \in [p-1]$, we define $\sigma v$ as the vector $(\sigma v[1], \dots, \sigma v[n])$ where each $v[i] \in \F_p$ for $i \in [n]$ and the operation $\sigma v[i]$ is the result of the multiplication operation of the field $\F_p$.

\begin{algorithm}
\caption{$\CoeffsTest{k}$}
\label{alg: random coeffs test}
\begin{algorithmic}[1]
\Ensure $k \in \N$
\Require Query access to a function $f: \F_p^n \to \F_p^r$.
\State Draw $k$ independent and uniformly random elements
$x_1, x_2, \dots, x_{k}$ from $\F_p^n$.
\State Draw $k$ independent and uniformly random coefficients $\sigma_1, \dots, \sigma_{k}$ from $[p-1]$.
\State Query $f(x_1), f(x_2), \dots, f(x_{k})$ and
$f(a)$, where  $a \gets \sum_{i\in[k]} \sigma_i x_i.$
\State\label{step:coeffs-condition}\textbf{Accept} if $\sum_{i\in [k]} \sigma_i f(x_i) = f(a)$; otherwise, \textbf{reject}.
\end{algorithmic}
\end{algorithm}

$\CoeffsTest{}$ accepts all homomorphisms from $\F_p^n \to \F_p^r$. Next we look at the soundness of this test in \Cref{thm:prime-field-soundness-based-on-eps}.

\begin{theorem}
\label{thm:prime-field-soundness-based-on-eps}
    For $k\in\N$, all $\eps \in (0,1)$, all primes $p$ and $n,r \in \N$, and functions $f:\F_p^n \to \F_p^r$ such that $f$ is $\eps$-far from being a group homomorphism,
        \begin{equation}
         \Pr{[\CoeffsTest{2k} \text{ rejects}]} \geq \min\set{\frac{(2k-3)\cdot \eps}{3}, \frac{1}{16}}.
    \end{equation}
\end{theorem}

The proof of \Cref{thm:prime-field-soundness-based-on-eps} closely parallels that of \Cref{thm:general-soundness-based-on-eps}, with each use of a sign sequence from $\set{+,-}^k$ replaced with a coefficient sequence from $[p-1]^k$. The rest of the argument is the same.

Next, we prove \Cref{thm:prime-fields-p=q}, following the case-based strategy of \cite{AroraKM25}, as in the proof of \Cref{thm:general_G_H}.
Define the parameter $m := 4 \lceil \log_p t + 15/\eps \rceil + 12$. When $m \leq n/4$, we use \Cref{alg: Online erasure resilent coeffs test}, a lift of the online erasure-resilient tester of \cite{BenEliezerKMR24} to prime fields. Otherwise, we use the sample-based $\GeneratedSubgroupTest$ (\Cref{alg:generated subgroup test}). Our main contribution here is again in the analysis of the distribution of the element $y$ in \Cref{alg: unpredictable coeffs test} where we need to handle random coefficients along with random elements for any group $G$.

\paragraph{Case 1: If $m \leq n/4$.} In this case, we use \Cref{alg: Online erasure resilent coeffs test} which internally runs \Cref{alg: unpredictable coeffs test} a constant number of times.

\begin{algorithm}
\caption{Unpredictable Random Coefficients Test}
\label{alg: unpredictable coeffs test}
\begin{algorithmic}[1]
\Ensure $m \in \N$, $m$ even
\Require Query access to $f: \F_p^n \to \F_p^r$ via $t$-erasure oracle. 
\State Draw $m$ independent and uniformly random elements $x_1, x_2, \dots, x_m$ from $\F_p^n$.
\label{line:online-erasure-resilient-coeffs-test-line-3}
\State Query $f(x_1), f(x_2), \dots, f(x_m)$. \label{line:online-erasure-resilient-coeffs-test-line-4}
\State Draw independent and uniformly random coefficients $\sigma_j$ for all $j \in [m]$ from $[p-1]$.
\State Let $S$ be a uniformly random subset of $[m]$ of size $m/2$.
\State Query $f(y)$ where  $y \gets \sum_{j\in S} \sigma_j x_j$ \label{line:online-erasure-resilient-coeffs-test-line-7}
\State \textbf{Accept} if $\sum_{j \in S} \sigma_j f(x_j) = f(y)$; otherwise \textbf{reject}.
\end{algorithmic}
\end{algorithm}

\begin{algorithm}
\caption{Online Erasure-Resilient Random Coefficients Test}
\label{alg: Online erasure resilent coeffs test}
\begin{algorithmic}[1]
\Ensure $\eps \in (0,1)$, $t \in \N$
\Require Query access to $f: \F_p^n \to \F_p^r$ via $t$-erasure oracle.
\State Set $m = 4 \lceil \log_p t + 15/\eps \rceil + 12$ 
\label{line:online-erasure-resilient-coeffs-test-line-1}
\For{$i \in [48]$ } \label{line:online-erasure-resilient-coeffs-test-line-2}
\State Run \Cref{alg: unpredictable coeffs test} with parameter $m$ and independent random coins with query access to $f$.
\State \textbf{Reject} if \Cref{alg: unpredictable coeffs test} rejects.
\EndFor
\State \textbf{Accept}
\end{algorithmic}
\end{algorithm}

We first prove \Cref{lem: random points linear independence} that shows that if we choose $n/2$ random elements from $\F_p^n$, then they are linearly independent with high probability.

\begin{lemma}
\label{lem: random points linear independence}
For all $n \in \N$, all primes $p$, let $v_1, \dots, v_{\lceil n/2 \rceil}$ be vectors chosen independently and uniformly at random from $\F_p^n$. Then with probability at least $ 1 -  \frac{2}{p^{n/2}} $, the vectors $ v_1, \dots, v_{\lceil n/2 \rceil} $ are linearly independent.
\end{lemma}

\begin{proof}
Let $B_1$ be the bad event that $v_1 = 0$. Then $\Pr{[B_1]} = 1/p^n$. For $i\geq 2$, let $B_i$ be the bad event that $v_{i}$ lies in $\mathrm{span}(v_1, \dots, v_{i-1})$ conditioned on the fact that $v_1, \dots v_{i-1}$ are linearly independent. Then $\Pr{[B_i]} = p^{i-1}/p^n$ since $v_1, \dots v_{i-1}$ span a $i-1$ dimensional subspace of $\F_p^n$ when they are linearly independent. The vectors $v_1, \dots v_{\lceil n/2 \rceil }$ are linearly independent when none of the bad events $B_1, \dots B_{\lceil n/2 \rceil}$ happen. By a union bound, the probability that at least one of the bad events happen is at most
\[
\Pr[B_i \text{ is true for at least one } i ] \leq \sum_{i=1}^{\lceil n/2 \rceil}\Pr{[B_i]} = \sum_{i=1}^{\lceil n/2 \rceil} \frac{p^{i-1}}{p^n}.
\]
Since this is a geometric sum, we get
\[
\sum_{i=1}^{\lceil n/2 \rceil} \frac{p^{i-1}}{p^n} = \frac{1}{p^n}\Big(\sum_{i=1}^{\lceil n/2 \rceil} p^{i-1} \Big) = \frac{1}{p^n} \cdot \Big( \frac{p^{\lceil n/2 \rceil} - 1}{p-1}\Big) \leq \frac{ 2 \cdot p^{n/2+1}}{p^{n+1}} = \frac{2}{p^{n/2}},
\]
 where the third inequality follows from $\frac{1}{p-1}\leq \frac{2}{p}$ for all $p\geq 2$ and $\lceil n/2 \rceil \leq n/2 + 1$.

Thus, with probability at least $ 1 - 2/p^{n/2} $, none of the bad events occur, and $ v_1, \dots, v_{\lceil n/2 \rceil} $ are linearly independent.
\end{proof} 

Next, \Cref{lem: unpredictable random coeffs completeness soundness} shows that in the case $k = m/4$, \Cref{alg: unpredictable coeffs test} has the same completeness and soundness as  $\CoeffsTest{2k}$. The proof is the same as the proof of \Cref{lem: unpredictable random signs completeness soundness}.

\begin{lemma}
\label{lem: unpredictable random coeffs completeness soundness}

For all $m \in \N$ divisible by 4, the Unpredictable Random Coefficients Test stated in \Cref{alg: unpredictable coeffs test} with parameter $m$ has the same completeness and soundness guarantees as $\CoeffsTest{2k}$ where $k = m/4$.
\end{lemma}

Next we will show that  \Cref{alg: unpredictable coeffs test} apart from simulating the \CoeffsTestInText also has some unpredictability. To do that, we first prove \Cref{lem: prime Dx is flat with high prob} which shows that the distribution of the last query $y$ in \Cref{alg: unpredictable coeffs test} is flat with high probability.

\begin{lemma}
\label{lem: prime Dx is flat with high prob}
If $m \leq n/4$, then the distribution of the last query $y$ in \Cref{alg: unpredictable coeffs test} which depends on the $m$ initial queries $X$, denoted as $\cD^X$ is $1/\binom{m}{m/2}(p-1)^{m/2}$-flat except with probability $2/p^m$ over the randomness of $X$.
\end{lemma}
\begin{proof}[Proof of \Cref{lem: prime Dx is flat with high prob}]
Once \Cref{alg: unpredictable coeffs test} chooses $X$ for this iteration, the distribution of $y$, $\cD^X$ is determined by the random choices of the coefficients $\overline{\sigma}$ and the subset $S$. Let $\cD^X$ be the distribution of $y$ conditioned on $X$ and $\cD^X_{S}$ be distribution conditioned on $X$ and the choice of the subset $S \subseteq [m]$ of size $m/2$. Since the subset is chosen uniformly at random, $\cD^X$ is a uniform mixture of the distributions $\cD^X_S$ for all $S \subseteq [m]$ of size $m/2$. Thus we can write 

    \[ \cD^X = \frac{1}{\binom{m}{m/2}} \sum_{\substack{S \subseteq [m]\\ \card{S} = m/2}} \cD^X_{S}.\]

\begin{definition}[Support of $\cD^X_{S}$]
   For all $X \in (\F_p^n)^m$ and for all subsets $S \subseteq [m]$ of size $m/2$, let $Y^X_{S}$ be the support of the distribution $\cD^X_{S}$, which is an affine space consisting of $(p-1)^{m/2}$ elements spanned by $X$ using all possible coefficient sequences of length $m/2$. Note that $Y^X_{S}$ is an affine space with $(p-1)^{m/2}$ elements and not a subspace with $p^{m/2}$ elements as the coefficients are from $[p-1]$ and can't be zero.
\end{definition}

Next we show that if $X$ is linearly independent and $S_1 \neq S_2$ then the supports $Y^X_{S_1}$ and $Y^X_{S_2}$ do not intersect.

\begin{claim}
     \label{clm: prime-fields-support-no-intersect}
        If the set $X$ is linearly independent, then for all subsets $S_1, S_2\subseteq [m]$ of size $m/2$ such that $S_1 \neq S_2$, the supports $Y^X_{S_1}$ and $Y^X_{S_2}$ do not intersect. 
     \end{claim}
     \begin{proof}[Proof of \Cref{clm: prime-fields-support-no-intersect}]
         Without loss of generality assume that there exists $j \in S_1 \setminus S_2$. For the sake of contradiction, assume that there exists an element $v$ which is present both in $Y^X_{S_1}$ and $Y^X_{S_2}$. Then $v$ can be written as two separate linear combinations, one with elements from $\set{x_i}_{i \in S_1}$ and the other with elements from $\set{x_i}_{i \in S_2}$ where the coefficients are from $[p-1]$:

         \[\sum_{i \in S_1} \alpha_i x_i  = v = \sum_{i \in S_2} \beta_i x_i \]
         \begin{equation}
         \label{eq: linear-dependence-of-X}
             \implies \sum_{i \in S_1 \cap S_2} (\alpha_i - \beta_i) x_i + \sum_{i \in S_1 \setminus S_2} \alpha_i x_i + \sum_{i \in S_2 \setminus S_2} (-\beta_i) x_i = 0,
         \end{equation}
         where $-\beta$ is the inverse of $\beta$ in $\F_p$. If $X$ is linearly independent, then the set $\set{x_i}_{i \in S_1 \cup S_2}$ is also linearly independent. Also, since $j \in S_1 \setminus S_2$, the coefficient of $x_j$ in \eqref{eq: linear-dependence-of-X} is not zero. This means that \eqref{eq: linear-dependence-of-X} gives an equation involving linearly independent vectors and scalars not all zero which add up to zero. However this is a contradiction, as this implies that the vectors in $\set{x_i}_{i \in S_1 \cup S_2}$ are linearly dependent. Thus, our initial assumption is wrong and the supports $Y^X_{S_1}$ and $Y^X_{S_2}$ do not intersect.
     \end{proof}

Now, we show that the distribution $\cD^X$ is $1/\binom{m}{m/2}(p-1)^{m/2}$-flat with probability at most $2/p^m$. If $X$ is linearly independent, then using \Cref{clm: prime-fields-support-no-intersect} we see that every element $z$ from $\F_p^n$ can only be present in at most one $Y^X_{S}$ for some subset $S$. Since the coefficient sequence $\overline{\sigma}$ is chosen independently and uniformly at random, the probability that $z$ is chosen under the distribution $\cD^X_{S}$ is at most $\frac{1}{(p-1)^{m/2}}$, which implies that the probability of $z$ being chosen under $\cD^X$ is at most $\frac{1}{{\binom{m}{m/2}} (p-1)^{m/2}}$ for all $z \in \F_p^n$.

 Let $B$ be the event that $\Pr_{y \sim \cD^X}{[y=z]} > \frac{1}{{\binom{m}{m/2}}(p-1)^{m/2}}$ from some element $z \in \F_p^n$. If $X$ is linearly independent, then the event $B$ cannot occur, thus $\Pr_{X}[B]$ is at most the probability of $X$ being linearly dependent. By \Cref{lem: random points linear independence}, since $m \leq n/4 < n/2$, we get
    \[\Pr_{X}[B] \leq \Pr_{X}{[X \text{ is linearly dependent}]} \leq \frac{2}{p^{n/2}} \leq \frac{2}{p^m}.\]
This completes the proof of \Cref{lem: prime Dx is flat with high prob}.
\end{proof}

Now, we are ready to show that \Cref{alg: unpredictable coeffs test} is an unpredictable tester in \Cref{lem: unpredictable coeffs test unpredictability}.

\begin{lemma}
\label{lem: unpredictable coeffs test unpredictability}
If $m \leq n/4$, then \Cref{alg: unpredictable coeffs test} with parameter $m$ is 
$$\Big (m+1,\frac{m}{p^m} + \frac{1}{\binom{m}{m/2}(p-1)^{m/2}}, \frac{2}{p^m} \Big )\text{-unpredictable.}$$
\end{lemma}
\begin{proof}
Firstly, note the \Cref{alg: unpredictable coeffs test} makes exactly $m+1$ queries. Consider all queries $x_i$ for $i \in [m]$ i.e., except the last one. From line 1 of \Cref{alg: unpredictable coeffs test}, we know that irrespective of the previous $i-1$ queries, $x_i$ is chosen independently and uniformly at random from $\F_p^n$. Thus, the distribution of $x_i$ is the uniform distribution over $\F_p^n$ which is $1/p^n$-flat with probability 1. Thus $\alpha_i = 1/p^n$ and $\beta_i = 0$ for all $i \in [m].$

Next we analyze what happens for the last query $y$. Let $\cD^X$ be the final distribution of $y$ conditioned on $X$. Using \Cref{lem: prime Dx is flat with high prob} we know that $\cD^X$ is $1/\binom{m}{m/2}(p-1)^{m/2}$-flat except with probability $2/p^m$ over the randomness of $X$. Thus $\alpha_{m+1} = 1/\binom{m}{m/2}(p-1)^{m/2}$ and $\beta_{m+1} = 2/p^m$.  Thus \Cref{alg: unpredictable coeffs test} is unpredictable with $$\sum_{ i \in [m+1]} \alpha_i = \frac{m}{p^n} + \frac{1}{\binom{m}{m/2}(p-1)^{m/2}} \leq \frac{m}{p^m} + \frac{1}{\binom{m}{m/2}(p-1)^{m/2}} = \alpha \text{ and } \sum_{ i \in [m+1]} \beta_i = \frac{2}{p^m} = \beta.$$
using the fact $m \leq n/4 \leq n \implies p^m  \leq p^n \implies 1/p^n \leq 1/p^m$.
\end{proof}

Finally, we show that \Cref{alg: Online erasure resilent coeffs test} satisfies \Cref{thm:prime-fields-p=q} for case 1, when $m \leq n/4$. 

\begin{proof}[Proof of \Cref{thm:prime-fields-p=q} when $m \leq n/4$]

The proof of this is similar to that of \Cref{thm:general_G_H} when $2^m \leq \card{G}^{1/4}$. We will use \Cref{lem: generic amplification} to show that \Cref{alg: Online erasure resilent coeffs test} is an online erasure resilient tester. Let $m$ be as defined in line 1 of \Cref{alg: Online erasure resilent coeffs test}, i.e., $m = 4 \lceil \log_p t + 15/\eps \rceil + 12$.

Let \Cref{alg: unpredictable coeffs test} be the base test needed in the statement of \Cref{lem: generic amplification}. Using \Cref{lem: unpredictable random coeffs completeness soundness}, we know that \Cref{alg: unpredictable coeffs test} with parameter $m$ has the same completeness and soundness as $\CoeffsTest{2k}$ where $k = m/4$.  Thus, by \Cref{thm:prime-field-soundness-based-on-eps}, \Cref{alg: unpredictable coeffs test} accepts all homomorphisms in $\mathsf{HOM}(\F_p^n,\F_p^r)$ and rejects functions that are $\eps$-far from $\mathsf{HOM}(\F_p^n,\F_p^r)$ with probability at least $\min\set{\frac{(m-6)\cdot \eps}{6}, \frac{1}{16}} \geq 1/16$ as 

\[ \frac{(m-6) \cdot \eps}{6}  = \frac{\big(4 \lceil \log_p t + 15/\eps \rceil + 6 \big) \cdot \eps}{6} > \frac{(1/\eps) \cdot \eps}{6} = 1/6 > 1/16.\]
Thus, the quantity $p_w(\eps) = 1/16$ as needed for using \Cref{lem: generic amplification}. Using \Cref{lem: unpredictable coeffs test unpredictability}, we know that \Cref{alg: unpredictable coeffs test} is $(q, \alpha, \beta)$-unpredictable where $$q = m+1, \alpha = \frac{m}{p^m} + \frac{1}{\binom{m}{m/2}(p-1)^{m/2}}, \beta = \frac{2}{p^m}.$$

To apply \Cref{lem: generic amplification}, we need the following condition (inequality \eqref{eq: amplification-condition} in the statemnt of \Cref{lem: generic amplification}) to be satisfied.
\[\alpha qt \cdot \Big \lceil \frac{3}{p_w(\eps)} \Big \rceil + \beta \le \frac{p_w(\eps)}{12} \]

Now we show that the above condition is satisfied by the values $\alpha, \beta, p_w(\eps)$ that we derived above.
\begin{align*}
        \alpha qt \cdot \Big \lceil \frac{3}{p_w(\eps)} \Big \rceil + \beta &= \Big( \frac{m}{p^m} + \frac{1}{\binom{m}{m/2}(p-1)^{m/2}}\Big) \cdot (m+1)t \cdot 48 + \frac{2}{p^m}\\
        & \leq \frac{48t(m+1)m}{p^m} + \frac{48t(m+1)}{\binom{m}{m/2}(p-1)^{m/2}} + \frac{2}{p^m} \\
        & \leq \frac{48t(m+1)m}{p^m} + \frac{48t(m+1)\sqrt{2m}}{2^m (p-1)^{m/2}} + \frac{2}{p^m} \leq \frac{256 t(m+1)m}{2^m(p-1)^{m/2}}.
    \end{align*}

where we use ${\binom{m}{m/2}} \geq 2^m/\sqrt{2m}$ in the second inequality and for the third inequality, we use $\sqrt{2m} \leq m$ for all $m>2$ and $2^m(p-1)^{m/2} \leq p^m$ for all $m$.

Now using $m = 4 \lceil \log_p t + 15/\eps \rceil + 12$, we get $ t \leq p^{m/4 - 15/\eps - 3} \leq p^{m/4 - 18}$. Thus,
\begin{equation}
\label{eq: prime-fields-erasure-prob}
    \frac{256 \cdot t(m+1)m}{2^m(p-1)^{m/2}} \leq \frac{256 \cdot (m+1)m \cdot p^{m/4}}{2^m(p-1)^{m/2} p^{18}}.
\end{equation}

 \begin{itemize}
        \item If $p=2$, the above quantity in RHS of \eqref{eq: prime-fields-erasure-prob} reduces to $$\displaystyle{\frac{256 \cdot (m+1)m }{2^{3m/4 +18}}} \leq \frac{256 \cdot 3}{2^{18}} \leq \frac{3}{1024} \leq \frac{1}{12 \cdot 16} = \frac{p_w(\eps)}{12}$$ where we use $\frac{(m+1)m}{2^{\frac{3m}{4}}} \leq 3$ for all $m$. 

        \item  Otherwise, if $p>2$, then $p \leq (p-1)^2 \implies p^{m/4} \leq (p-1)^{m/2}$ for all $m$. Thus, the RHS of \eqref{eq: prime-fields-erasure-prob} at most  $$\displaystyle{\frac{256 \cdot (m+1)m}{2^m p^{18}}} \leq \frac{256 \cdot 3}{2^{18}} \leq \frac{3}{1024} \leq \frac{1}{12 \cdot 16} = \frac{p_w(\eps)}{12}$$ where we use $\frac{(m+1)m}{2^{m}} \leq \frac{(m+1)m}{2^{\frac{3m}{4}}} \leq 3$ for all $m$. 
    \end{itemize}

Thus using \Cref{lem: generic amplification}, we conclude that \Cref{alg: Online erasure resilent coeffs test} that runs $ \lceil 3/p_w(\eps) \rceil = 48$ independent iterations of \Cref{alg: unpredictable coeffs test} with $m = 4 \lceil \log_p t + 15/\eps \rceil + 12$ is an online manipulation resilient tester.
\end{proof}

\paragraph{Case 2: If $m > n/4$.} In this case, we use $\GeneratedSubgroupTest$ (\Cref{alg:generated subgroup test}) from \Cref{sec:group-specif}.

\begin{proof}[Proof of \Cref{thm:prime-fields-p=q} when $m > n/4$]
The size of the minimum size of any generating set for $\F_p^n$, i.e., $d(\F_p^n) = n$. This follows from the fact that any set of less than $n$ vectors from $\F_p^n$ cannot span the whole space, however the $n$ unit vectors along each coordinate can span $\F_p^n$. 

Since the additive group structure of $\F_p^n$ is Abelian the expected number of random elements to generate $\F_p^n$, i.e., $E(\F_p^n)$ is at most $n+3$ \cite{pomerance2002expected}.
Using \Cref{thm:group-specific-sample-based-erasure-resilient-result}, we get that \Cref{alg:generated subgroup test} is a $t$-online manipulation-resilient tester for $t \leq c \cdot \min \set{\eps^2, 1/n^2} p^n$ for some constant $c >0$ which uses $O(1/\eps + n)$ which is at most $O(1/\eps + m) = O(1/\eps + \log_p t)$ queries. 
\end{proof}

Next we prove \Cref{thm:prime-fields-lower-bound} which shows that \Cref{thm:prime-fields-p=q} is essentially tight when the range is $\F_p$. The proof of \Cref{thm:prime-fields-lower-bound} closely follows the proof structure of \cite[Theorem 1.3]{KalemajRV23} which proves a lower bound of $\Omega(\log t)$ on the query complexity of online erasure-resilient testing of linearity of functions from $\F_2^n \to \F_2$. The proof uses Yao's minimiax principle for the online-erasures model from \cite[Corollary 9.4]{KalemajRV23} stated in \Cref{thm: yao-minimax-online-erasure-model}.

\begin{definition}[$\cD$-view]
Given a $q$-query deterministic algorithm $\mathcal{A}$, let $a(x)$ be the string of $q$ answers that $\mathcal{A}$ receives when making queries to input object $x$. For a distribution $\mathcal{D}$ and adversarial strategy $\mathcal{S}$, let $\mathcal{D}$-view be the distribution on strings $a(x)$ where the input object $x$ is sampled from $\mathcal{D}$ and accessed via a $t$-online-erasure oracle using strategy $\mathcal{S}$.
\end{definition}

\begin{theorem}[Yao's minimax principle for the online-erasures model \cite{KalemajRV23}]
\label{thm: yao-minimax-online-erasure-model}
To prove a lower bound $q$ on the worst-case query complexity of online-erasure-resilient randomized algorithms for a promise decision problem, it is enough to give

\begin{itemize}
    \item a randomized adversarial strategy $S$,

    \item a distribution $\cD^+$ on positive instances of size $n$,

    \item a distribution $\cD^-$ on instances of size $n$ that are negative with probability at least $6/7$
\end{itemize}
such that the statistical distance between $\cD^+$-view and $\cD^-$-view is at most $1/6$ for every deterministic $q$-query algorithm that accesses its input via an online-erasure oracle using strategy $S$.
\end{theorem}

\begin{proof}[Proof of \Cref{thm:prime-fields-lower-bound}] Let $\cD^+$ be the uniform distribution over all homomorphisms from $\F_p^n \to \F_p$ and $\cD^-$ be the uniform distribution over all functions from $\F_p^n \to \F_p$.

\begin{claim}
\label{clm: random-function-is-from-linear}
The probability that a function $f$ sampled from $\cD^-$ is $\frac{p-1}{2p}$-far from being a homomorphism is at least $6/7$. 
\end{claim}
\begin{proof}
Fix a homomorphism $h:\F_p^n \to \F_p $. For each $x \in \F_p^n $, define the indicator random variable $Z_x = 1$ if $f(x) \ne h(x)$ and 0 otherwise. Since $f$ is sampled from $\cD^-$, the value $f(x)$ is chosen uniformly at random over $\F_p$, independently for each $x$. Thus,  \[\E{Z_x} = \Pr{[f(x) \ne h(x)]} = \frac{p-1}{p}.\]
The distance between $f$ and $h$ is defined as 
$\displaystyle\mathsf{dist}(f, h) = \frac{1}{p^n} \sum_{x \in \mathbb{F}_p^n} Z_x.$
Applying Hoeffding's inequality to the sum of $ p^n $ independent bounded random variables $Z_x \in \set{0,1} $ with $\delta = \frac{p-1}{2p}$, we get
\[
\Pr\left[ \mathsf{dist}(f, h) \leq \frac{p-1}{2p}\right] =
\Pr\left[ \mathsf{dist}(f, h) \leq \frac{p-1}{p} - \delta \right] \leq \exp(-2\delta^2 p^n) = \exp(-\frac{(p-1)^2p^{n-2}}{2}).
\]
Note that the number of group homomorphisms $ h: \mathbb{F}_p^n \to \mathbb{F}_p $ is exactly $ p^n $, since each such map corresponds to a linear function determined by $ n $ coefficients in $ \mathbb{F}_p $. Taking a union bound over all $ p^n $ homomorphisms, we obtain
\[
\Pr\left[\exists h \text{ such that } \mathsf{dist}(f, h) \leq \frac{p-1}{2p} \right] \leq p^n \cdot \exp \Big(-\frac{(p-1)^2p^{n-2}}{2}\Big).
\]
This is exponentially small in $ p^n $. For large enough $n$, this probability is at most $\frac 17$.
\end{proof}

We fix the following strategy for a $t$-online-erasure adversary: after each query is answered, erase a subset $T$ of $t$ elements from the subspace spanned by the queries so far. If there are multiple possibilities for $T$, break ties arbitrarily. 
Notice that if at most $\log_p t$ queries are made, then the adversary can erase all the elements from the subspace spanned by the queries, except the queries themselves.

Let $\cA$ be a deterministic algorithm that makes $q \leq \log_p t$ queries. Let $\cP^+$ and $\cP^-$ be two random processes that interact with $\cA$ while $\cA$ makes it queries. For each query, if the query was erased by the adversary, then both $\cP^+$ and $\cP^-$ return $\perp$, otherwise they return a uniformly random value from $\F_p$.  From the above description, it is clear that the distribution over the transcripts when $\cA$ interacts with $\cP^+$ is identical to the distribution over the transcripts when $\cA$ interacts with $\cP^-$.

Next we describe how $\cP^+$ and $\cP^-$ assigns values to the \emph{undisclosed} (i.e., unqueried or erased) elements. 
\begin{itemize}
    \item \textbf{$\cP^-$ outputs a function sampled from $\cD^-$}. Here $\cP^-$ assigns a uniformly random value from $\F_p$ to each undisclosed element. Thus, $\cP^-$ outputs a random function from~$\cD^-$. 

    \item \textbf{$\cP^+$ outputs a function sampled from $\cD^+$}. Let $V$ be the set of queries among the $q$ queries $\cA$ made, where it received a value from $\F_p$ not $\perp$. The adversary's strategy ensures that the queries in $V$ are linearly independent. Then $\cP^+$ completes $V$ to a basis $B$ that spans $\F_p^n$. Let the elements of $B$ be denoted as $x_1, \dots, x_n$. Next $\cP^+$ assigns a uniformly random value from $\F_p$ to each element in $B \setminus V$. Since $B$ spans $\F_p^n$, all elements in $\F_p^n$ can be written as a linear combination of elements from $B$. For each undisclosed element $y$ in $ \F_p^n \setminus B$, let $y = \sum_{i \in [n]} a_i x_i$ where $a_i \in \F_p$, then $\cP^+$ assigns the value $\sum_{i \in [n]} a_i f(x_i)$ to $y$. 
    Thus, $\cP^+$ outputs a random function from $\cD^+$. 
\end{itemize}

Thus, we conclude that $\cD^+$-view is equal to $\cD^-$-view. All functions in the support of $\cD^+$ are homomorphisms, whereas, by \Cref{clm: random-function-is-from-linear}, a function $f \sim \cD^-$ is $\frac{p-1}{2p}$-far from any homomorphism with probability $\frac 6 7$. Applying \Cref{thm: yao-minimax-online-erasure-model} gives the desired lower bound of $\Omega(\log_p t)$.
\end{proof}

\bibliographystyle{alpha}
\bibliography{refs}

\begin{appendix}

\section{Miscellaneous proofs}

\begin{claim}
\label{clm: sum is random}
    For all $k \in \N$, all sign sequences $\overline{\sigma} = (\sigma_1, \sigma_2, \dots, \sigma_k) \in \set{+,-}^k$ and all finite groups $G$, if $x_1, \dots, x_k$ are chosen independently and uniformly at random from $G$, then the sum $\sum_{i \in [k]} \sigma_i x_i$ is a uniformly distributed element in $G$.
\end{claim}
\begin{proof}
    Let $a$ be an element in $G$. 
        Then
\begin{align}
    \label{eq: sum is random}
    \Pr\Big[\sum_{i \in [k]} \sigma_i x_i = a\Big] = \Exp_{x_1, \dots, x_{k-1}}{\Big[\Pr_{x_k}\Big[\sum_{i \in [k]} \sigma_i x_i = a\Big]\Big]}.
    \end{align}
Fix the values of the first $k-1$ elements, so the only randomness is in $x_k$. The condition
$\sum_{i \in [k]} \sigma_i x_i = a$
is equivalent to
$ \sigma_k x_k = \sum_{i \in [k]}\sigma_{k-i}x_{k-i} + a$. Therefore, there is only one choice of $x_k$ out of $|G|$ that satisfies the equality: the element $\sum_{i \in [k]}\sigma_{k-i}x_{k-i} + a$ if $\sigma_k=+$ and the inverse of this element if $\sigma_k=-$.
 Thus, the probability term inside the expectation in \eqref{eq: sum is random} is $1/\card{G}$.
Taking the expectation over these terms yields
$\Pr\left[\sum_{i \in [k]} \sigma_i x_i = a\right] = 1/|G|.$

\end{proof}

The following fact is folklore. For the sake of completeness, we add a short proof of it .
\begin{fact}
\label{fact:binomial rv is even}
Let $ X \sim \mathrm{Bin}(n, p) $. Then the probability that $ X $ is even is given by:
\[
\Pr(X \text{ is even}) = \frac{1 + (1 - 2p)^n}{2}.
\]
\end{fact}

\begin{proof}
Let $ q = 1 - p $. Consider the function $f(z)$ defined as,
\[
f(z) = (q + pz)^n = \sum_{k=0}^{n}  \binom{n}{k} p^k q^{n - k} z^k.
\]
Evaluating at $ z = 1 $ and $ z = -1 $, we get:
\[
f(1) = (p + q)^n = 1, \quad f(-1) = (q - p)^n = (1 - 2p)^n.
\]
Observe that
\begin{align*}
f(1) + f(-1) &= \sum_{k=0}^{n} \binom{n}{k} p^k q^{n - k} (1 + (-1)^k) \\
&= 2 \sum_{\substack{k=0 \\ k \text{ even}}}^{n} \binom{n}{k} p^k q^{n - k}.
\end{align*}
Therefore,
\[
\Pr(X \text{ is even}) = \sum_{\substack{k=0 \\ k \text{ even}}}^{n} \binom{n}{k} p^k q^{n - k} = \frac{f(1) + f(-1)}{2} = \frac{1 + (1 - 2p)^n}{2}.\qedhere
\]
\end{proof}

\begin{claim}
\label{clm: upper bound on reimann zeta}
    Let $\zeta(x)$ be the Reimann Zeta function. Then $ \zeta(x) \leq 1 + \frac{1}{2^{x-1} - 1}$ for all $x \geq 2$.      Moreover, for all $\beta \in (0,1)$, if $x = 1 + \log_2 (1+ 1/\beta)$, then $\zeta(x) \leq 1 + \beta$.
\end{claim}
\begin{proof}
    By the definition of the Reimann Zeta function,
\[
\zeta(x) = \sum_{n=1}^\infty \frac{1}{n^x} = 1 + \sum_{n=2}^\infty \frac{1}{n^x}.
\]
We partition the tail sum $ \sum_{n=2}^\infty \frac{1}{n^x} $ into dyadic intervals:
\[
[2,4), [4,8), [8,16), \ldots, [2^k, 2^{k+1}) \text{ and so on}.
\]
Each interval $ [2^k, 2^{k+1}) $ contains $ 2^k $ integers, and for all integers $ n \in [2^k, 2^{k+1}) $, we have:
\[
n \geq 2^k \quad \Rightarrow \quad \frac{1}{n^x} \leq \frac{1}{(2^k)^x} = 2^{-kx}.
\]
Therefore, the contribution from each block is bounded by
\[
\sum_{n=2^k}^{2^{k+1}-1} \frac{1}{n^x} \leq 2^k \cdot 2^{-kx} = 2^{k(1 - x)}.
\]
Summing over $ k \geq 1 $,
\[
\sum_{n=2}^\infty \frac{1}{n^x} \leq \sum_{k=1}^\infty 2^{k(1 - x)} = \sum_{k=1}^\infty \left(2^{1 - x}\right)^k = \frac{2^{1 - x}}{1 - 2^{1 - x}} = \frac{1}{2^{x-1} - 1}.
\]
The geometric series converges as the ratio between successive terms, $2^{1-x} < 1$ as $x \geq 2$. Thus,
\[
\zeta(x) = 1 + \sum_{n=2}^\infty \frac{1}{n^x} \leq 1 + \frac{1}{2^{x-1} - 1}.
\]
Substituting $t = 1 + \log_2 (1+ 1/\beta)$, we get
$\zeta(t) \leq 1 + \frac{1}{2^{\log_2(1+ 1/\beta)} - 1} = 1 + \beta$.
\end{proof}

\end{appendix}

\end{document}